\documentclass[a4paper, 11pt, fleqn]{article}

\usepackage[english]{babel}
\usepackage[utf8x]{inputenc}
\usepackage[T1]{fontenc}
\usepackage{authblk}
\usepackage{natbib}

\usepackage[a4paper,top=2.6cm,bottom=2.6cm,left=2.85cm,right=2.85cm]{geometry}

\usepackage{amsmath, amssymb, amsfonts, amsthm, mathtools, mathrsfs, dsfont}
\usepackage{graphicx}
\usepackage{multirow, multicol}
\usepackage{booktabs}%
\usepackage{xcolor}
\usepackage[algo2e, ruled]{algorithm2e}
\usepackage{enumitem}
\usepackage[colorlinks=true, allcolors=blue]{hyperref}
\usepackage[title]{appendix}%
\usepackage{textcomp}
\usepackage{manyfoot}
\mathtoolsset{showonlyrefs}
\graphicspath{{figures/}}

\newtheorem{lemma}{Lemma}
\newtheorem{example}{Example}
\newtheorem{remark}{Remark}

\newcommand*{\limite}[3][]{\overset{#1}{\underset{#2\rightarrow #3}{\longrightarrow}}}
\DeclareMathOperator*{\supp}{supp}
\renewcommand{\mid}{\;\ifnum\currentgrouptype=16 \middle\fi|\;}
\providecommand{\keywords}[1]{\small\textbf{{Keywords.}} #1}

\begin{document}

\title{Simulating signed mixtures}
\date{Authors contributed equally to this work.}

\author[1,3]{Christian P. Robert}
%\equalcont{Both authors contributed equally to this work.}
\author[1,2]{Julien Stoehr}

\affil[1]{CEREMADE, Université Paris-Dauphine, Université PSL, CNRS, 75016 Paris, France}
\affil[2]{UMR MIA Paris-Saclay, INRAE, AgroParisTech, Université Paris-Saclay, 91120 Palaiseau, France}
\affil[3]{Department of Statistics, University of Warwick, Coventry, CV4 7AL, UK}

\maketitle

\begin{abstract}
Simulating mixtures of distributions with both positive and negative (signed) weights proves a challenge as standard simulation algorithms prove inefficient in handling the negative weights. In particular, the natural representation of mixture random variates as being associated with latent component indicators is no longer available. We propose an exact accept-reject algorithm for the general case of finite signed mixtures that relies on optimally pairing positive and negative components and designing a stratified sampling scheme on these pairs. We analyze the performances of our approach, relative to the inverse cdf approach, since the cdf of the targeted distribution remains available for signed mixtures of common distributions.

\vspace{\baselineskip}
\noindent\keywords{Acceptance-rejection algorithm, Mixtures, Signed mixtures, Simulation, Inverse cdf, Quantile function}
\end{abstract}

\section{Introduction}

\subsection{Signed mixtures}
Mixture distributions \citep{titterington:smith:makov:1985} abound in the statistical literature as a ubiquitous tool to represent inhomogeneous populations and to enlarge the collection of common distributions \citep[see, e.g., ][]{maclachlan01}. The density function of such distributions
writes as a linear combination of $K>1$
base density functions $f_k$, $1 \leq k \leq K$, namely
\begin{equation}\label{eq:orig_mix}
\sum_{k=1}^K \omega_k f_k\,,\quad\text{where}\quad
\begin{cases}
\sum_{k=1}^K \omega_k=1\,, \\
\omega_1,\ldots,\omega_K>0\,.
\end{cases}
\end{equation}
It is then straightforward to simulate from \eqref{eq:orig_mix} when the component density functions $f_k$ are themselves easily simulated: a component index $1\le k\le K$ is selected with probability $\omega_k$ and a realization from $f_k$ is then produced. This simplicity explains why many simulation methods in the literature exploit an intermediary mixture construction
to speed up the production of pseudo-random samples from more challenging distributions. For instance, \citet[Section XIV.4.5]{devroye85} points out that unimodal distributions can be written as countable mixtures of uniform distributions. Similarly, mixture distributions are often selected as proposals in MCMC algorithms \citep{robert:casella:2004,cappe:douc:guillin:marin:robert:2008}.

In this paper, we consider the more challenging setting of {\em signed} mixtures, namely the case when the mixture weights, $\omega_1,\ldots,\omega_K$, in \eqref{eq:orig_mix} are {\em signed}, that is, when some of the $\omega_k$'s are {\em negative}. 
We define a
{\em signed mixture density} as a linear combination of $P\ge 1$ positively weighted density functions $f_k$ and of $N\ge 1$ negative weighted density functions $g_k$, with the constraint that the combination remains a properly defined probability density. Denoting by $\supp(h)$ the support of a real-valued function $h$, this implies that the joint support of the positively weighted density functions must contain the joint support of the negatively weighted density functions, namely
\begin{equation*}
%\label{eqn:support} in case needed
\bigcup_{1\le k\le N}\supp(g_k) \subseteq \bigcup_{1\le k\le P} \supp(f_k)\,.
\end{equation*}
A {\em signed mixture density} $m$ thus writes as
\begin{equation}
m = \sum_{k=1}^P \omega_k^+ f_k - \sum_{k=1}^N \omega_k^- g_k\,,\quad\text{with}\quad
\begin{cases}
\sum_{k=1}^{P} \omega_k^+ - \sum_{k=1}^{N} \omega_k^-=1\,, \\
\omega_1^+,\ldots,\omega_{P}^+, \omega_1^-,\ldots,\omega_{N}^- >0
\end{cases}
\label{eqn:signed-mixt}
\end{equation}
and the constraint that $m$ is a non-negative function. (This property is indeed sufficient to ensure $m$ is a probability density.) 

\subsection{Simulation from signed mixtures}
Density functions expressed as \eqref{eqn:signed-mixt} present a significant challenge when considering the generic issue of simulating them and we are not aware of existing solutions to this problem. Indeed, a na{\"\i}ve solution consists in first simulating realizations from the associated mixture of positive weight components, which writes as
\begin{equation}
\label{eqn:m+}
\sum_{k=1}^{P} \omega_k^+ f_k \Big/ {\sum_{k=1}^{P} \omega_k^+}\,,
\end{equation}
when renormalized, and then using an accept-reject post-processing step \citep{bignami:dematteis:1971,devroye85,robert:casella:2004} that subsamples values among these simulations. While formally correct, this approach may prove highly inefficient since the marginal probability of acceptance 
\begin{equation*}
{\displaystyle 1}\Big/{\displaystyle \sum_{k=1}^{P} \omega_k^+}
\end{equation*}
can be arbitrarily close to zero. Furthermore, as already observed by \citet{devroye85}, checking for the acceptance condition is potentially costly if $K = P + N$ is large. The intuition behind the computational inefficiency of the standard accept-reject algorithm is that simulating from the positive weight components $f_k$ is not necessarily producing values within regions of high probability for the actual distribution \eqref{eqn:signed-mixt}. Indeed, it is possible that the negative weight components $\omega_k^-g_k$ remove most of the mass attached to the $\omega_k^+f_k$'s. Therefore high probability regions for $m$ have no reasons in general to co\"incide with high probability regions for all of the $f_k$'s. By the same argument, resorting to an accept-reject method based on the mixture of the {\em negative} weight components is similarly inefficient. In addition, the so-called {\em series method} proposed by \citet[Section IV.5]{devroye85} is not well-suited for this target density since it requires a manageable functional upper bound on \eqref{eqn:signed-mixt}, namely one that can be simulated.  Efficient alternatives are thus necessary and we herewith propose a generic solution.

\subsection{Prevalence of signed mixtures}
The motivation for considering {\em signed} mixtures and their simulation is many-fold. Besides approximations proposed for simulation reasons \citep{devroye85}, signed mixtures appear in series representations of density functions \citep{beaulieu:1990,delaigle:hall:2010,hubalek:kuznetsov:2011} or as flexible modelling tools \citep[][the later using the denomination of {\em subtractive mixtures}]{zhang:zhang:2005,mueller:2012,kroese:etal:2019,loconte:etal:2024}. The kernel conditional density estimators constructed by \cite{schuster:etal:2020} open the possibility of signed mixtures, while \cite{polson2024negativeprobability} connect signed mixtures with the notion of negative probability appearing in quantum theory. \cite{loconte:etal:2024} provides further references about the use of signed mixtures in machine learning, optimization, and signal processing.

Specific examples of probability density functions represented as signed series include the Raab-Green distribution, the Kolmogorov-Smirnov (test) distribution, and the Erd\"os-Kac distribution \citep[IV.5]{devroye85}. For instance, \cite{elston:glasbey:1989} study the special case of Exponential signed mixtures 
\begin{equation*}
\sum_{k=1}^P \omega_k^+ \mathcal E(\lambda_k^+) - \sum_{k=1}^N \omega_k^- \mathcal E(\lambda_k^-)\,,\quad\text{such that}\quad\lambda_1^+,\ldots,\lambda_N^->0\,,
\end{equation*}
by exploiting a connection with {\em generalised Erlang distributions}, whose density functions are themselves linear combinations of multiple Exponential density functions with some 
negative coefficients. However, the complexity of their approach is of order $\mathcal{O}(2^{P+N})$, which calls for a more efficient alternative. Similarly, the bivariate Exponential distribution proposed by \cite{gumbel:1960} is a signed mixture of bivariate 
Gamma distributions, whose efficient simulation is of direct interest in extreme value theory and copula representations. 

A primary remark is that, when the sole purpose of the simulation is the approximation of integrals related with \eqref{eqn:signed-mixt}, the negativity of some weights is not directly a hindrance since
\begin{equation*}
\int h(x) m(x)\mathrm{d}x =
\int h(x) \sum_{k=1}^{P} \omega_k^+ f_k(x) \mathrm{d}x-\int h(x) \sum_{k=1}^{N} \omega_k^- g_k(x) \mathrm{d}x
\end{equation*}
holds. It thus suffices to produce simulations from both positive weight and negative weight mixtures. However, this solution may prove inefficient when $K$ is large and when the supports of the positive and negative density functions strongly overlap. (The above decomposition also explains why the cdf of \eqref{eqn:signed-mixt} can be computed, when considering $h(y)=\mathbb{I}_{(-\infty,x)}(y)$.)

\subsection{A core decomposition for signed mixtures}This paper approaches the simulation from a signed mixture \eqref{eqn:signed-mixt} by rewriting $m(\cdot)$ 
using a non-unique decomposition of the positive and negative weights and a rearrangement into three terms
\footnote{In order to lighten the notational burden, \eqref{eqn:mixtby} reuses notations such as $f_k$ and $g_k$, with no exact correspondence with those found in \eqref{eqn:signed-mixt}. For instance, the number $K$ of components in \eqref{eqn:mixtby} may be equal to the product $P\times N$. % in \eqref{eqn:signed-mixt}.
}, 
\begin{equation}
m=\sum_{k=1}^K \lambda_k \{a_k f_k - g_k\}
+\sum_{i = 1}^P r_i f_i - \sum_{j = 1}^N s_j g_j,
\label{eqn:mixtby}
\quad\text{such that}\quad
\begin{cases}
\lambda_1, \ldots, \lambda_K > 0\,,\\
r_1, \ldots, r_P\geq 0\,,\\ 
s_1, \ldots, s_N \geq 0\,,
\end{cases}
\end{equation}
and for all $1 \leq k \leq K$,
%\begin{equation*}
$\inf_{x\in\mathbb R} \{a_kf_k(x) -g_k(x)\} \geq 0 \,$.
%\end{equation*}

\begin{example}
Denoting  $\varphi(\cdot;\mu,\sigma^2)$ the probability density function of the Normal distribution with mean $\mu$ and variance $\sigma^2$,  consider the Normal signed mixture
\begin{equation*}
    m(x) = 2 \varphi(x;0,1) + 1.8 \varphi(x;0.5,1) 
    - \varphi(x;0.25, 0.25) - 0.8 \varphi(x;0.75, 0.16).
\end{equation*}
A valid decomposition of this signed mixture as \eqref{eqn:mixtby} is, for instance,
\begin{align*}
m(x) & = 
0.8\{2.25\varphi(x;0,1) - \varphi(x;0.75, 0.16)\}
+ 0.9\{2\varphi(x;0.5,1) - \varphi(x;0.25, 0.25)\}\\
& \hspace{2em}+ 0.2 \varphi(x;0,1) - 0.1\varphi(x;0.25,0.25)\,.
\end{align*}
since it can be easily checked that the first two signed mixtures are positive functions.
\end{example}

The construction and optimization of \eqref{eqn:mixtby} will be conducted in Section \ref{sec:pairs}. The argument behind this representation \eqref{eqn:mixtby} is that a generic signed mixture \eqref{eqn:mixtby} can always be written\footnote{In some special cases from the literature, the signed mixtures already come paired as in \eqref{eqn:mixtby}.} as a mixture of $K$ two-component signed mixtures, the $\{a_k f_k - g_k\}$'s, plus potential positive and negative residual terms. Simulation-wise, the appeal attached to \eqref{eqn:mixtby} is that those residuals have low probability mass and hence most of the draws from \eqref{eqn:mixtby} correspond to the first sum in \eqref{eqn:mixtby}, whose simulation is straightforward. Indeed, this simulation proceeds by first selecting at random a component index $k$ with probability proportional to $\lambda_k$ and second generating from this component density $\{a_kf_k-g_k\}/(a_k-1)$ by a na\"ive accept-reject approach when $a_k$ is small enough, or by a more elaborate accept-reject method that is developed below, otherwise.

The plan of the paper is as follows. In Section \ref{sec:2cpt}, we construct a specific simulation method for two-component signed mixtures. Section \ref{sec:pairs} details how the pairing decomposition of \eqref{eqn:mixtby} is chosen. Section \ref{sec:compar} contains numerical experiments that compare different approaches of this simulation challenge.
Technical details are postponed till Appendices \ref{app:two-comp}, \ref{app:pairing}, \ref{app:inv-cdf}, \ref{app:rand-mix}. Appendix \ref{app:two-comp} specifically focuses on two-component signed mixtures and elaborates in details examples of the signed mixtures of two Normal or two Gamma distributions.

\subsection{Notations and conventions}

Throughout the paper, we do not distinguish between the measures and their associated density functions. In what follows, the probability density function (pdf) of the signed mixture (with respect to the Lebesgue measure) is denoted by $m$. The positive and negative weight components are consistently referred to as $f_i$, $1 \leq i \leq P$ and $g_j$, $1\leq j\leq N$, respectively, with indices omitted when there is no ambiguity. The positive part $m^+$ of a signed mixture corresponds to the mixture \eqref{eqn:m+}.

For a set $D \subseteq \mathbb{R}^d$, we denote by $\nu(D)$ the probability that a random variable with density $\nu$ belongs to $D$ and by $\vert D \vert$ the volume of $D$, that is
\begin{equation*}
\vert D \vert = \int \mathds{1}_{D}(x)\mathrm{d}x.
\end{equation*}
We always assume that the cumulative distribution functions (cdf) of positive and negative weight components can be computed everywhere so that
\begin{equation*}
m(D) = \sum_{k = 1}^P \omega_k^+f_k(D) - \sum_{k = 1}^N \omega_k^- g_k(D)
\end{equation*}
is available.

\section{Two-component signed mixtures}
\label{sec:2cpt}

In this section, we address the specific case of a signed mixture with a sole positive and a sole negative weights. Given two distinct probability density functions $f$ and $g$ such that $\supp(g)\subseteq\supp(f)$, and
\begin{equation}
a^\star = \sup_{x\in\supp(f)} {g(x)}\big/{f(x)} < +\infty\,,
\label{eqn:cond-pair}
\end{equation}
a {\em two-component signed mixture} of $f$ and $g$ is defined as
\begin{equation}
m = \frac{af - g}{a - 1}\,,\quad\text{with}\quad a\geq a^\star. 
\label{eqn:def-pair}
\end{equation}
Condition \eqref{eqn:cond-pair} ensures that $m$ stands as a proper probability density when $g$ 
has tails that are dominated by those of the positive component $f$. Note that $a^\star > 1$ as in generic accept-reject settings (see Appendix \ref{app:two-comp}, Lemma \ref{lemma:a-star}). The limiting case $a = a^\star$ corresponds to the minimal positive weight required to compensate the negative weight component, i.e., when the density function $m$ reaches zero at some point of its support or asymptotically.

\paragraph*{Vanilla sampling scheme} As mentioned earlier, a natural, albeit na{\"\i}ve, method for sampling from  \eqref{eqn:def-pair} consists in an accept-reject algorithm with proposed values generated from the distribution $f$. Since
\begin{equation*}
\sup_{x\in\supp(f)} {m(x)}\big/{f(x)} = \frac{a}{a-1}\,,
\end{equation*}
the proposed values $x$ are accepted with probability\begin{equation*}
\frac{af(x)-g(x)}{af(x)}.
\end{equation*}
The average acceptance probability is equal to $(a-1)/a$, which makes the approach inefficient when $a-1$ is near zero, i.e., when $f$ and $g$ are quite similar. 

%It is possible to elaborate further on that opportunity to use samples from $f$. A two-component signed mixture can be written as a mixture of $f$ and the two-component signed mixture $m^\star$ corresponding to the limiting case of $a = a^\star$ in \eqref{eqn:def-pair}, namely
%\begin{align}
%\frac{a^\star - 1}{a - 1}m^\star + \frac{a-a^\star}{a - 1}f\,,\quad\text{where}\quad m^\star = \frac{a^\star f - g}{a^\star - 1}.
%\label{eqn:negmix-pair-opt}
%\end{align}
%Reframing the problem in this way shows that sampling from $m$ can be achieved by sampling from $f$ but performing the accept reject step solely for samples whose latent state is associated with $m^\star$, that is accepting a simulated value $x$ from $f$ with probability
%\begin{align*}
%\frac{a - a^\star}{a - 1} + \frac{a^\star - 1}{a - 1}\frac{a^\star f(x)-g(x)}{a^\star f(x)}.
%\end{align*}
%The advantage is that it leads to an overall better-expected acceptance probability (see Lemma \ref{lemma:opt-ar-pair} in Appendix)
%\begin{align}
%\rho(a) = \frac{a - a^\star}{a - 1} + \frac{(a^\star - 1)^2}{a^\star (a - 1)}.
%\end{align}

\paragraph*{Stratified sampling scheme}
We can instead construct an alternative accept-reject scheme based on a piecewise upper bound on \eqref{eqn:def-pair} towards yielding a higher acceptance probability on average. For this purpose, consider a partition
$(D_0, \ldots, D_n)$ of $\supp(f)$ with the convention that $D_0$ contains the tails of $f$ and potential subsets where the density $f$ is unbounded. We assume that upper and lower bounds on both $f$ and $g$, over all remaining elements of the partition, $D_i$, $1 \leq i \leq n$, exist and can be easily computed, so that the terms
\begin{equation*}
    h_i = a\sup_{x\in D_i} f(x) - \inf_{x\in D_i} g(x),\quad 1 \leq i \leq n
\end{equation*}
are available. These terms yield a rough upper bound on $m$ on each $D_i$ that can obviously be improved in the specific situation when direct access to the supremum of $m$ on $D_i$ is available. Tails are treated separately. Indeed, since the tail dominating component is necessarily attached to the positive part, $af$ can then be used as an upper bound of $m$ on $D_0$. The partition is therefore providing a direct and different upper bound on \eqref{eqn:def-pair}, that is, for all $x\in\supp(f)$,
\begin{equation*}
m(x) \leq \frac{1}{a-1}\left\lbrace
af(x)\mathds{1}_{D_0}(x) + \sum_{i = 1}^{n} h_i \mathds{1}_{D_i}(x)
\right\rbrace. 
\end{equation*} 
This dominating function can be normalised into a proposal density, towards a novel accept-reject algorithm since sampling from this proposal is straightforward. 
Since this proposal is a special instance of a mixture distribution, a possible sampling strategy is as follows: one picks a partition element, that is a (latent) component index, at random according to the vector of probabilities of its components 
\begin{equation*}
    \varrho = (a f(D_0),h_1\vert D_1\vert,\ldots,h_n\vert D_n\vert)\left/ \left\lbrace af(D_0)+\sum_{i=1}^n h_i\vert D_i\vert \right\rbrace\right.
\end{equation*}
and then simulates from $f$ restricted to $D_0$ 
or uniformly on $D_i$, $1 \leq i \leq n$, respectively. Note that $D_0$ can be further decomposed towards making the simulation of the truncated distribution manageable in practice.

This strategy is however computationally sub-optimal since, in order to obtain one draw from $m$, it requires to sample on average (see Appendix \ref{sec:proof-algo-1})
\begin{equation*}
M = \frac{a}{a-1} f(D_0) + \frac{1}{a-1} \sum_{i = 1}^{n} h_i \vert D_i\vert
\end{equation*}
(latent) component index variables, while only one is needed. 
Hence, we switch to a more efficient strategy. We follow a stratified sampling method (as detailed in Algorithm \ref{algo:ar2}) that takes advantage of the partition structure as well as of the availability of the cdf of \eqref{eqn:def-pair}. First, the procedure selects a partition element $D_k$ according to the signed mixture $m$, i.e., it draws the partition index $k$ according to the Multinomial distribution
$\mathcal{M}(m(D_0), \ldots, m(D_n))$. 
Once the partition element $k$ is generated, the procedure samples a realization from the distribution $m$ restricted to $D_k$. This is achieved by performing an accept-reject step that uses as proposal a truncated distribution if $D_0$ was picked, and a uniform distribution otherwise.

The stratified approach reduces the total number of simulated random variables, even though the average acceptance probability of Algorithm \ref{algo:ar2} remains the same as for the na\"\i ve accept-reject algorithm, that is $1/M$. 
While this modification might seem accessory, as the simulation method within a partition element remains unchanged, it offers the key advantage, of
cutting the average computational budget of the proposition step from $3M$ to $1 + 2M$ random variables, in contrast to the initial strategy. 

\begin{algorithm2e}[t]
\label{algo:ar2}
\caption{Accept-reject method for two-component signed mixtures}
\KwIn{Partition $D_0, \ldots, D_n$, upper bounds $h_1, \ldots, h_n$.}
\vspace{\baselineskip}
{\bf sample} $k$ from $\mathcal{M}(m(D_0),\ldots,m(D_n))$;

\Repeat{accepting}{
\If{k = 0}{
{\bf sample} $x$ from $f$ truncated to $D_0$\;
{\bf accept} $x$ with probability $\{af(x) - g(x)\}/\{af(x)\}$\;
} \Else{
{\bf sample} $x$ uniformly on $D_k$\;
{\bf accept} $x$ with probability $\{af(x) - g(x)\}/h_k$\;
}
}
\end{algorithm2e}

Obviously, the initial partition can easily be refined into smaller sets towards controlling the overall acceptance probability, as stated by the following result (Appendix \ref{sec:proof:lem:accept-proba} details its proof).
\begin{lemma}
\label{lem:accept-proba}
Let $\delta \in (1-1/a, 1)$ 
and $\varepsilon \in [0, (1 - \delta)/\delta)$. If
\begin{equation}
\label{hyp:d0}
g(D_0) = \frac{(a - 1)\{1 - \delta(\varepsilon + 1)\}}{\delta},
\end{equation}
then there exists a partition $(D_1, \ldots, D_{n_{\varepsilon}})$ of $\supp(f)\setminus D_0$ such that the average acceptance probability of Algorithm \ref{algo:ar2} is greater than $\delta$.
\end{lemma}

The constraint on parameters $\delta$ and $\varepsilon$ ensures that (i) $g(D_0)\in(0, 1)$ and (ii) Algorithm \ref{algo:ar2} outperforms the vanilla method in terms of average acceptance probability. The result also
provides a heuristic on how to build the partition to achieve this.
If we aim at an overall acceptance probability of $\delta$, we first build $D_0$ so it satisfies \eqref{hyp:d0} for a user-specified tolerance $\varepsilon$. 
The purpose of this threshold is twofold: it informs on the average acceptance probability for a countably infinite partition, namely $\delta/(1-\delta\varepsilon)$ (see Appendix \ref{sec:proof:lem:accept-proba}), and, more interestingly, on the largest error possible when approximating $1 - m(D_0)$ by the upper Riemann sum
\begin{equation*}
    \frac{1}{a-1} \sum_{i = 1}^{n_{\varepsilon}} h_i \vert D_i\vert.
\end{equation*}
Note that this error can be larger than $1$ when the targeted acceptance probability is lower than $0.5$. With a positive tolerance level\footnote{Note that setting $\epsilon=0$ serves no practical purpose, as it means having a countably infinite partition.},
the stratified approach leaves room for improving performances. The mass of $D_0$ with respect to $g$ decreases with $\varepsilon$. Choosing $\varepsilon$ close to $(1 - \delta)/\delta$ allows for larger errors but this requires to partition a larger domain. 
Conversely, choosing $\varepsilon$ close to $0$ involves partitioning a smaller domain but requires a possibly larger cardinal of the partition. 
Once $D_0$ is set, we recursively divide $\supp(f)\setminus D_0$ to find a suitable $(D_1, \ldots, D_{n_{\varepsilon}})$, that is till the upper Riemann sum approximates $1 - m(D_0)$ with an error less than $\varepsilon$:
\begin{equation*}
\frac{1}{a-1} \sum_{i = 1}^{n_{\varepsilon}} h_i \vert D_i\vert - 1 + m(D_0) \leq \varepsilon
\quad\text{i.e.}\quad
\frac{1}{a-1} \sum_{i = 1}^{n_{\varepsilon}} h_i \vert D_i\vert - \frac{1}{\delta} + \frac{af(D_0)}{(a - 1)} < 0.
\end{equation*}
Various recursive processes can be used to achieve this stopping rule. In Section \ref{sec:compar}, we started with a partition based on equally spaced points, and we then recursively refined every partition element $D_i$ for which
\begin{equation*}
    \frac{1}{a-1}h_i\vert D_i \vert - m(D_i) > \frac{\varepsilon}{n_\varepsilon}.
\end{equation*}
For a practical implementation of building such a partition, we refer the reader to the example in Appendix \ref{app:hist}.

\section{Pairing mechanism}
\label{sec:pairs}

For a generic signed mixture \eqref{eqn:signed-mixt}, it is rarely the case that the density $m$ naturally appears in the format \eqref{eqn:mixtby}. We thus propose a method to construct a pairing of positive and negative (weight) components and a residual mixture towards a representation of the mixture as \eqref{eqn:mixtby} that improves the average acceptance probability.

For a given signed mixture \eqref{eqn:signed-mixt}, denote $E$ the set of all acceptable pairs of positive and negative weight component indices, i.e., such that we can define a two-component signed mixture from the associated density functions, namely
\begin{equation*}
E = 
\left\lbrace
(i, j) , 
\begin{aligned}
\,1 \leq i \leq P
\\ %\, 
1\leq j \leq N 
\end{aligned}
\mid
\supp(g_j) \subseteq \supp(f_i)
\quad\text{and}\quad
a_{ij}^\star = \sup_{x\in\supp(f_i)} \frac{g_j(x)}{f_i(x)} < +\infty
\right\rbrace.
\end{equation*}
The set $E$ is always non-empty since, otherwise, the signed mixture $m$ could not constitute a proper probability density. Subsequently, $E_i^+$ and $E_j^-$ will denote the sets of pairs that contain the positive component $i$ and the negative component $j$, respectively.

A {\em pairing} refers to a set of two-component signed mixtures that can be constructed from mixture $m$, and
is defined as a subset $F\subset E$, and a collection of weights $(\omega_{ij}^+, \omega_{ij}^-)_{(i,j)\in F}$ that satisfy the following constraints
\begin{align}
\forall (i,j) \in F, & \quad\omega_{ij}^+ - a_{ij}^\star\omega_{ij}^- \geq 0, \label{eqn:valid-pair}
\\
\forall i, \,\,1\leq i \leq P,  &\quad\sum_{(i, j)\in E_i^+\cap F} \omega_{ij} ^+ \leq \omega_i^+, \label{eqn:const-wp}
\\
\forall j, \,\,1\leq j \leq N,  &\quad\sum_{(i, j)\in E_j^-\cap F} \omega_{ij} ^- \leq \omega_j^-.  \label{eqn:const-wn}
\end{align}
The constraint \eqref{eqn:valid-pair} ensures that the weights associated with the pair $(i, j)$ define a two-component signed mixture that is positive everywhere. 
Constraints \eqref{eqn:const-wp} and \eqref{eqn:const-wn} guarantee 
that when we gather the two-component signed mixtures, the overall weight does not exceed the total weight in $m$.

A pairing is associated with a residual mixture
\begin{equation*}
\sum_{i = 1}^P r_i f_i - \sum_{j=1}^N s_j g_j,
\quad\text{where}\quad
r_i = \omega_i^+ - \sum_{(i, j)\in E_i^+\cap F} \omega_{ij} ^+
\quad\text{and}\quad
s_j = \omega_j^- - \sum_{(i, j)\in E_j^-\cap F} \omega_{ij} ^-.
\end{equation*}
The decomposition of $m$ associated with the pairing thus writes as
\begin{equation}
\label{eqn:negmix-decomp}
\sum_{(i,j) \in F} (\omega_{ij}^+ f_i - \omega_{ij}^- g_j) + \sum_{i = 1}^P r_i f_i - \sum_{j=1}^N s_j g_j.
\end{equation}
Sampling from $m$ can hereby be achieved by proposing a sample from the mixture made of the two-component signed mixtures and of the positive weight components, namely
\begin{equation*}
\pi = \sum_{(i,j) \in F} \frac{\omega_{ij}^+ - \omega_{ij}^-}{C}
\left(
\frac{\omega_{ij}^+ f_i - \omega_{ij}^- g_j}{\omega_{ij}^+ - \omega_{ij}^-}
\right) 
+ \sum_{i = 1}^P \frac{r_i}{C} f_i,
\quad\text{where}\quad
C = \sum_{i = 1}^P \omega_i^+ - \sum_{(i,j)\in F}\omega_{ij}^-,
\end{equation*}
and by accepting the resulting simulation $x$ with probability
%\begin{equation*}
${m(x)}\big/{C\pi(x)}$.
%\end{equation*}
Sampling from $\pi$ proceeds as for any standard (unsigned) mixture distribution, albeit requiring an extra accept-reject step when sampling from the component of $\pi$ that corresponds to pairs $(i,j)\in F$ (see Algorithm \ref{algo:ar-gen}). 
\begin{algorithm2e}[t]
\label{algo:ar-gen}
\caption{Accept-reject method for general signed mixtures}
\KwIn{A pairing $F$.}
\vspace{\baselineskip}
{\bf compute} the vector of probabilities
$\text{prob} \propto \left(\omega_{ij}^+ - \omega_{ij}^-, r_\ell
\right)_{(i,j) \in F, 1\leq \ell \leq P}$\;
\Repeat{accepting}{
{\bf sample} $k$ according to prob\;
\If{k is associated with a pair $(i, j) \in F$}{
{\bf sample} $x$ from the two-component signed mixture with an accept-reject scheme\;
} \Else{
{\bf sample} $x$ from $f_k$\;
}
{\bf accept} $x$ with probability $m(x)/\{C\pi(x)\}$.
}
\end{algorithm2e}

If sampling each pair $(i,j)\in F$ relies on the vanilla approach, the overall procedure resumes to sampling by proposing from the mixture $m^+$ \eqref{eqn:m+}. Indeed, one sample from $m$ requires $C$ samples from $\pi$ on average, and to get one sample from $\pi$ we need to propose
\begin{equation*}
\sum_{(i,j) \in F} \frac{\omega_{ij}^+ - \omega_{ij}^-}{C}  \frac{\omega_{ij}^+}{\omega_{ij}^+ - \omega_{ij}^-} + \sum_{i = 1}^P \frac{r_i}{C} = \frac{1}{C} \sum_{i = 1}^P \omega_i^+
\end{equation*}
random variables on average. 
However, improving the acceptance probability to sample from at least one of the two-component signed mixtures involved in the decomposition \eqref{eqn:negmix-decomp} is enough to improve the performance of the sampling method, as shown by the following result (whose proof is given in Appendix \ref{sec:proof:lem:opt-pair}).\\

\begin{lemma}
\label{lem:opt-pair}
Consider $\delta\in(0,1)$ and a pairing $F$ for $m$. Assume that we sample from each pair $(i, j) \in F$, 
using
\begin{enumerate}[topsep = 5pt, itemsep = 5pt]
\item the vanilla sampling scheme if $(1 - \delta)\omega_{ij}^+ - \omega_{ij}^- \geq 0$, that is if the average acceptance probability of the vanilla scheme associated to the pair is larger than $\delta$,
\item a piecewise sampling scheme that guarantees an average acceptance probability greater than $\delta$, otherwise.
\end{enumerate}
Then Algorithm \ref{algo:ar-gen} requires on average less than
\begin{equation*}
\sum_{i = 1}^P \omega_i^{+}  + \frac{1}{\delta} \sum_{(i,j)\in F} \left\lbrace (1 - \delta)\omega_{ij}^+ - \omega_{ij}^- \right\rbrace \mathds{1}_{\{(1 - \delta)\omega_{ij}^+ - \omega_{ij}^- < 0\}}
\end{equation*}
proposed random variables to sample once from $m$.
\end{lemma}

A direct consequence of this result is to define the optimal pairing scheme (in terms of the number of proposed samples) as the one that minimizes the objective function
\begin{equation*}
\sum_{(i,j)\in F} \left\lbrace (1 - \delta)\omega_{ij}^+ - \omega_{ij}^- \right\rbrace \mathds{1}_{\{(1 - \delta)\omega_{ij}^+ - \omega_{ij}^- < 0\}}.
\end{equation*}
The solution to this optimization problem is equivalent to minimizing the objective function
\begin{equation}
\label{eqn:obj-simplex}
\sum_{(i,j)\in E} \left\lbrace (1 - \delta)\omega_{ij}^+ - \omega_{ij}^- \right\rbrace
\end{equation}
under the linear constraints \eqref{eqn:valid-pair}, \eqref{eqn:const-wp} and \eqref{eqn:const-wn} (refer to the justification in Appendix \ref{sec:app:simplex}). 
The optimal pairing solution can thus be found by an optimization algorithm targeting the above objective, such as the simplex method \citep{dantzig:1963}.

We stress that Algorithm \ref{algo:ar-gen} does not necessarily achieve an overall average acceptance probability of $\delta$ for the optimal pairing. Indeed, the average number of proposed random variables for the pairing writes as
\begin{equation*}
\sum_{i = 1}^P \omega_i^{+}  + \frac{1}{\delta} \sum_{(i,j)\in F} \left\lbrace (1 - \delta)\omega_{ij}^+ - \omega_{ij}^- \right\rbrace 
= \frac{1}{\delta} + \left(1 - \frac{1}{\delta}\right)\sum_{i = 1}^P r_i 
+ \frac{1}{\delta} \sum_{j = 1}^N s_j.
\end{equation*}
It is then lower than $1/\delta$ only when
\begin{equation*}
    \sum_{j = 1}^N s_j \leq (1 - \delta) \sum_{i = 1}^P r_i.
\end{equation*}
For instance, when an optimal pairing has no positive weight residuals, attaining exactly the targeted probability $\delta$ is then achieved solely if we have no negative residuals as well. Even though we control the acceptance probability when sampling from a pair, the reject step towards getting samples of $m$ by simulating from $\pi$ degrades the overall performances. Conversely, if there are no negative weight residuals,
Algorithm \ref{algo:ar-gen} achieves a higher acceptance probability than the user-specified rate $\delta$. This setting does not involve a reject step to get from $\pi$ to $m$. In that case, we do control sampling performances for each pair and each positive weight residual can be simulated exactly.

\section{Comparison experiments}
\label{sec:compar}

In this section, we examine the performances of three methods that return simulations from arbitrary signed mixture distributions $m$, namely
\begin{enumerate}
    \item the vanilla scheme corresponding to the accept-reject method based on the positive part of $m$,
    \item the stratified scheme we proposed for acceptance probabilities $\delta \in\{0.4, 0.6, 0.8\}$ and tolerance levels $\varepsilon \in \{0.1, 0.2, 0.5, 1\}$ compatible with $\delta$,
    \item a numerical inversion of the cumulative distribution function associated with $m$ for a precision of $10^{-10}$ (see Appendix \ref{app:inv-cdf}).
\end{enumerate}
Each method is run to get $n \in \{10, 10^2, 10^3, 10^4\}$ samples from $m$. For a given sample size $n$, we report the proportion $\widehat{\delta}_n$ of accepted proposed variables. Its theoretical value is denoted $\delta$ for both the vanilla and the stratified schemes. We also detail the relative efficiency of a method $A$ compared to a method $B$, defined as the ratio of the running time of $B$ by the running time of $A$. A relative efficiency larger than 1 indicates that $A$ outperforms $B$ in terms of computational budget.
We focus on the relative efficiency $\mathcal{R}_n$ of our method compared to the vanilla approach and the relative efficiency $\mathcal{Q}_n$ of accept-reject based methods compared to the numerical inversion of the cdf.

While the above construction is as generic as possible, we run the comparison on special instances of signed mixtures of exponential families distributions, namely $f_k$ and $g_k$ are both either Normal or Gamma distributions. Both families enjoy an explicit condition for \eqref{eqn:cond-pair} to hold and hence define a proper two-component signed mixture of $f_k$ and $g_k$ (see Appendix \ref{app:exp-fam}). We also provide details on how to build the partition $D_0, \ldots, D_{n_\varepsilon}$ for such a two-component signed mixture in Appendix \ref{app:hist}. For each family, we consider two kinds of numerical experiments.

\subsection{Alternating signed mixtures}

The first comparison is provided for a particular signed mixture that writes as the alternating sum 
\begin{equation}
\label{eqn:alt-sum}
    m \propto \sum_{k = 1}^K \left(\frac{a_k^\star}{a_k^\star - 1}\right)
    \left(
\frac{a_k^\star f_k - g_k}{a_k^\star - 1},
\right)
\end{equation}
where each term involves the two-component signed mixture \eqref{eqn:def-pair} of $f_k$ and $g_k$ for the minimal positive weight possible. Such a signed mixture exhibits a natural pairing structure where the weight of each pair in the overall signed mixture is inversely proportional to the average acceptance probability of the pair. There exists at least one solution to the optimisation problem that comes with no residual mixture.

\begin{figure}[t]%
\centering
\includegraphics[width=0.49\textwidth]{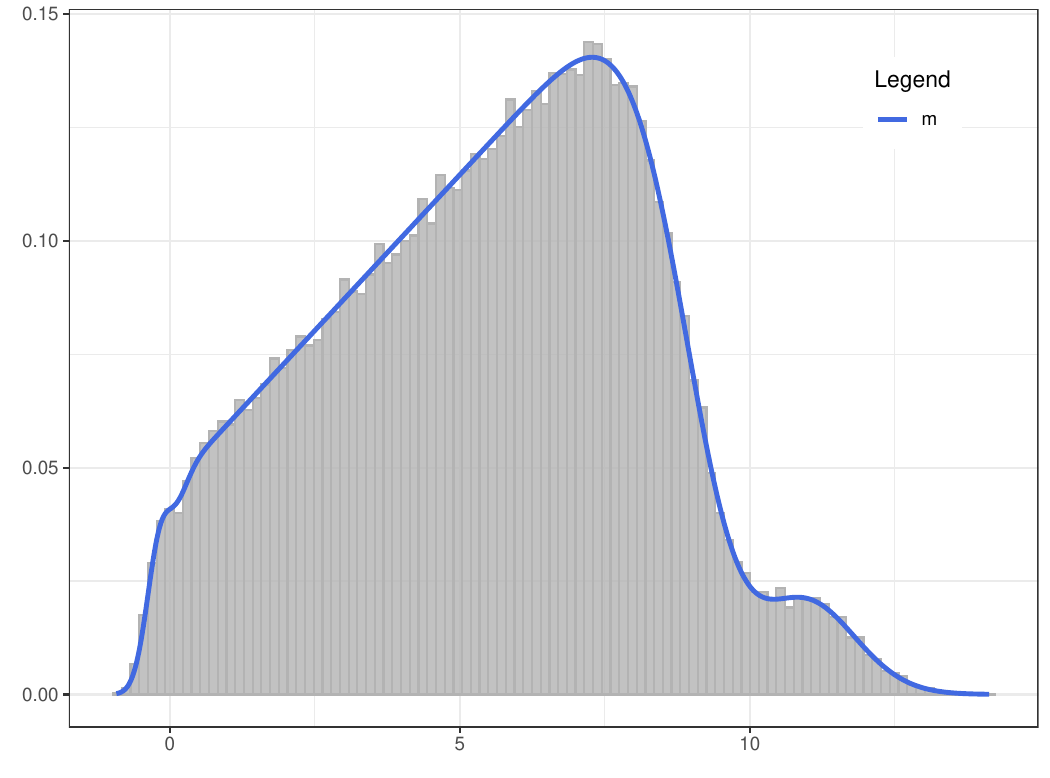}
\includegraphics[width=0.49\textwidth]{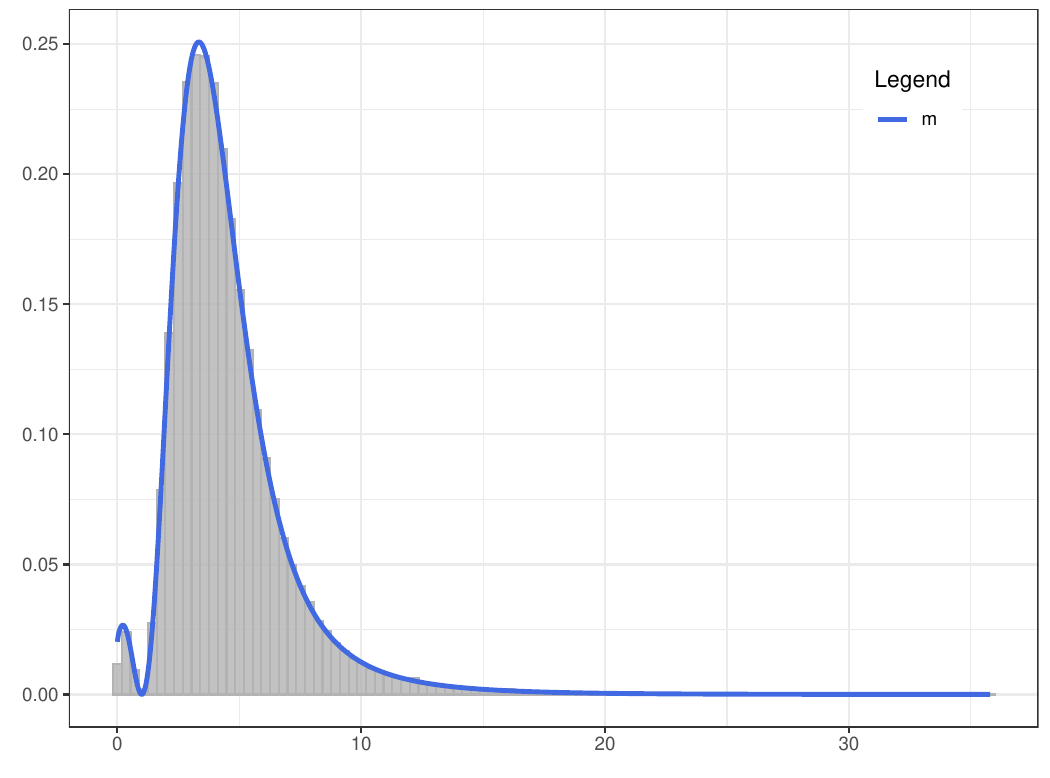}
\caption{Histogram of $10^5$ samples from an alternating signed mixture \eqref{eqn:alt-sum} with Normal distributions \eqref{eqn:norm-1} (left) and Gamma distributions \eqref{eqn:gam-1} (right).}
\label{fig:hist-examples}
\end{figure}

\paragraph*{Signed mixture of Normal distributions} A two-component Normal signed mixture can only be defined when the variance of the positive weight component is strictly greater than the variance of the negative one (see Appendix \ref{app:normal}).
We thus consider the signed mixture \eqref{eqn:alt-sum} with $K = 51$ and for all $k$, $1 \leq k \leq 51$, 
\begin{equation}
\label{eqn:norm-1}
\begin{aligned}
f_k & \equiv \mathcal{N}(\mu_k, \sigma_k^2)\\
g_k & \equiv \mathcal{N}(\mu_k + 0.01, (\sigma_k - 0.01)^2)
\end{aligned}
\quad\text{with}\quad
\begin{cases}
\mu_k = 0.2\, (k - 1),
\\
\sigma_k = 0.25 +  0.015\, (k - 1),
\end{cases}
\end{equation}
and
\begin{equation*}
a_k^\star = \frac{\sigma_k}{\sigma_k - 0.1} \exp\left\lbrace
\frac{0.01}{4\sigma_k - 0.02}
\right\rbrace.
\end{equation*}

\paragraph*{Signed mixture of Gamma distributions} A two-component Gamma signed mixture can only be defined when the shape and rate of the positive weight component are lower, respectively strictly lower, than the shape and rate of the negative one (see Appendix \ref{app:gamma}). As an example of the signed mixture \eqref{eqn:alt-sum}, we consider a setting with $K = 41$ and for all $k$, $1 \leq k \leq 41$, 
\begin{equation}
\label{eqn:gam-1}
\begin{aligned}
f_k & \equiv \Gamma(\alpha_k, \beta_k)\\
g_k & \equiv \Gamma(\alpha_k + 0.01, \beta_k + 0.01)
\end{aligned}
\quad\text{with}\quad
\begin{cases}
\alpha_k = 1 + 0.1\, (k - 1),
\\
\beta_k = 0.25 +  0.04375\,(k - 1),
\end{cases}
\end{equation}
and
\begin{equation*}
a_k^\star = \frac{\Gamma(\alpha_k)}{\Gamma(\alpha_k + 0.1(k - 1))}\left(\frac{\beta_k}{\beta_k + 0.01}\right)^{0.01(k - 1)}\exp\left\lbrace 0.01(1 - k)\right\rbrace.
\end{equation*}

\begin{table}[!b]
\caption{Sampling performances for alternating signed mixtures \eqref{eqn:alt-sum} of Normal distributions \eqref{eqn:norm-1} and Gamma distributions \eqref{eqn:gam-1}.\label{tab:alt-sum}}
{\centering
\begin{tabular*}{\textwidth}{@{\extracolsep\fill}lcccc|ccc|cc|c}
\hline
\hline
& \multicolumn{9}{@{}c@{}|}{Stratified} & \multirow{3}{*}{Vanilla}
\\
\cline{2-10}
% $\delta$ & \multicolumn{4}{@{}c@{}|}{0.400} & \multicolumn{3}{@{}c@{}|}{0.600} & \multicolumn{2}{@{}c@{}|}{0.800} \\
$\varepsilon$  & $0.1$ & $0.2$ & $0.5$ & $1.0$ & $0.1$ & $0.2$ & $0.5$ & $0.1$ & $0.2$ \\
\hline
\multicolumn{11}{@{}c@{}}{\textsc{Normal signed mixture}} \\
\hline
%$\delta$ & 0.400 & 0.400 & 0.400 & 0.400 & 0.600 & 0.600 & 0.600 & 0.800 & 0.800 & 0.018 \\
$\delta$ & \multicolumn{4}{@{}c@{}|}{0.4} & \multicolumn{3}{@{}c@{}|}{0.6} & \multicolumn{2}{@{}c@{}|}{0.8} & 0.018 \\
\hline
$\widehat{\delta}_{10}$ & 0.156 & 0.256 & 0.500 & 0.769 & 0.588 & 0.769 & 0.833 & 0.278 & 1.000 & 0.017\\
$\mathcal{R}_{10}$ & 1.876 & 2.136 & {\bf 2.153} & 2.125 & 1.604 & 1.720 & {\bf 1.738} & 1.128 & {\bf 1.208} & 1.000\\
$\mathcal{Q}_{10}$ & 2.892 & 3.293 & {\bf 3.319} & 3.276 & 2.473 & 2.652 & {\bf 2.680} & 1.738 & {\bf 1.862} & 1.542\\
\hline
$\widehat{\delta}_{10^2}$ & 0.407 & 0.592 & 0.417 & 0.633 & 0.389 & 0.637 & 0.645 & 0.820 & 0.714 & 0.018\\
$\mathcal{R}_{10^2}$ & 10.70 & 11.90 & 10.93 & {\bf 12.71} & 10.49 & {\bf 11.69} & 11.66 & 8.374 & {\bf 9.995} & 1.000\\
$\mathcal{Q}_{10^2}$ & 2.213 & 2.461 & 2.262 & {\bf 2.630} & 2.170 & {\bf 2.417} & 2.413 & 1.732 & {\bf 2.067} & 0.207\\
\hline
$\widehat{\delta}_{10^3}$ & 0.422 & 0.442 & 0.480 & 0.657 & 0.615 & 0.607 & 0.737 & 0.810 & 0.868 & 0.017\\
$\mathcal{R}_{10^3}$ & 26.05 & 27.25 & 27.76 & {\bf 27.94} & 11.39 & {\bf 19.94} & 16.08 & 22.25 & {\bf 23.70} & 1.000\\
$\mathcal{Q}_{10^3}$ & 4.521 & 4.730 & 4.817 & {\bf 4.849} & 1.976 & {\bf 3.461} & 2.791 & 3.861 & {\bf 4.114} & 0.174\\
\hline
$\widehat{\delta}_{10^4}$ & 0.411 & 0.426 & 0.499 & 0.640 & 0.604 & 0.641 & 0.776 & 0.827 & 0.855 & 0.018\\
$\mathcal{R}_{10^4}$ & {\bf 38.05} & 36.51 & 37.02 & 37.04 & 37.59 & 38.65 & {\bf 38.84} & 36.96 & {\bf 37.96} & 1.000\\
$\mathcal{Q}_{10^4}$ & {\bf 6.233} & 5.980 & 6.063 & 6.068 & 6.157 & 6.330 & {\bf 6.362} & 6.053 & {\bf 6.218} & 0.164\\
\hline
\multicolumn{11}{@{}c@{}}{\textsc{Gamma signed mixture}}
\\
\hline
% $\delta$ & 0.400 & 0.400 & 0.400 & 0.400 & 0.600 & 0.600 & 0.600 & 0.800 & 0.800 & 0.008 \\
$\delta$ & \multicolumn{4}{@{}c@{}|}{0.4} & \multicolumn{3}{@{}c@{}|}{0.6} & \multicolumn{2}{@{}c@{}|}{0.8} & 0.008 \\
\hline
$\widehat{\delta}_{10}$ & 0.909 & 0.769 & 1.000 & 1.000 & 0.833 & 1.000 & 1.000 & 1.000 & 0.714 & 0.006\\
$\mathcal{R}_{10}$ & 0.674 & 0.715 & 0.783 & 0.790 & 0.426 & 0.428 & 0.428 & 0.386 & 0.471 & $\mathbf{1.000}$\\
$\mathcal{Q}_{10}$ & 1.731 & 1.836 & 2.010 & 2.028 & 1.095 & 1.100 & 1.098 & 0.990 & 1.209 & {\bf 2.568}\\
\hline
$\widehat{\delta}_{10^2}$ &  0.435 & 0.408 & 0.595 & 0.495 & 0.752 & 0.498 & 0.820 & 0.990 & 0.962 & 0.010\\
$\mathcal{R}_{10^2}$ & 2.606 & 2.422 & 3.240 & {\bf 3.744} & 2.029 & {\bf 2.109} & 1.606 & 1.740 & {\bf 2.178} & 1.000\\
$\mathcal{Q}_{10^2}$ & 1.613 & 1.499 & 2.005 & {\bf 2.317} & 1.256 & {\bf 1.305} & 0.994 & 1.077 & {\bf 1.348} & 0.619\\
\hline
$\widehat{\delta}_{10^3}$ & 0.418 & 0.385 & 0.487 & 0.648 & 0.646 & 0.520 & 0.640 & 0.884 & 0.858 & 0.009\\
$\mathcal{R}_{10^3}$ & 20.11 & 20.51 & 25.82 & {\bf 27.13} & 17.28 & 18.34 & {\bf 22.77} & 15.23 & {\bf 19.45} & 1.000\\
$\mathcal{Q}_{10^3}$ & 5.728 & 5.844 & 7.355 & {\bf 7.727} & 4.923 & 5.225 & {\bf 6.487} & 4.339 & {\bf 5.540} & 0.285\\
\hline
$\widehat{\delta}_{10^4}$ & 0.420 & 0.426 & 0.456 & 0.523 & 0.581 & 0.644 & 0.602 & 0.828 & 0.852 & 0.009\\
$\mathcal{R}_{10^4}$ & 54.03 & 56.36 & 70.68 & {\bf 90.18} & 61.65 & 64.92 & {\bf 77.01} & 66.44 & {\bf 75.97} & 1.000\\
$\mathcal{Q}_{10^4}$ & 13.44 & 14.02 & 17.59 & {\bf 22.44} & 15.34 & 16.15 & {\bf 19.16} & 16.53 & {\bf 18.90} & 0.249\\
\hline
\hline
\end{tabular*}
}
\footnotesize{$\delta$, $\widehat\delta_n$: theoretical average acceptance probability of the method and its estimated value for a $n$-sample. $\mathcal{R}_n, \mathcal{Q}_n$: relative efficiency for a $n$-sample of the sampling method compared respectively to the vanilla method and the numerical inversion of the cdf.}
\end{table}

\paragraph*{Comments}
Table \ref{tab:alt-sum} displays the results for both families. In both examples, the simplex method retrieves the natural pairing associated with the alternating sum form \eqref{eqn:alt-sum} for all $\delta\in\{0.4, 0.6, 0.8\}$. The stratified method overall outperforms both the vanilla method and the numerical inversion of the cdf, regardless of the selected acceptance probability $\delta$ and the tolerance level $\varepsilon$. Unless simulating a dozen variables, our method is between 1.6 and 90 times faster than the vanilla method while the reduction in computation time is smaller when compared to the numerical inverse of the cdf but can still go up to a factor 22.
In general, for a given acceptance probability $\delta$, increasing the tolerance level $\varepsilon$ results in a lower computational cost of our stratified method, supporting the hypothesis that integration error prevails when building the partition. Conversely, the higher the acceptance probability, the higher the cost of our method. This pattern directly results from the construction of the partition $D_0, \ldots, D_{n_\varepsilon}$, where a higher acceptance probability implies a larger domain to partition and a smaller tolerance requires finer partition elements. 
Lastly, the computational benefit increases with the number of variables simulated, as the cost of both the simplex method and the computation of the partition becomes negligible in front of the cost of sampling random variables.

\subsection{Randomly generated signed mixtures}

The second comparison is based on a collection of 2,800 randomly generated signed mixtures (see Appendix \ref{app:rand-mix}) with a wide range of variety from the number of components to the average acceptance probability of the vanilla method. Table \ref{tab:rep-rand-negmix} details the distribution of the models into 7 categories depending on the acceptance probability of the vanilla method. The aim was to have models with arbitrary low vanilla acceptance probability in order to challenge our approach in situations where the vanilla method may perform extremely poorly. Models considered also encompass a few components up to a hundred with varying proportions of positive and negative weight components, ensuring then real diversity in the complexity of models (see Figure \ref{fig:summary-rand-negmix}).

\begin{table}[!h]
    \caption{Repartition of the 2,800 randomly generated signed mixtures of Normal distributions and Gamma distributions according to the average acceptance probability $\delta$ of the vanilla accept-reject method.\label{tab:rep-rand-negmix}}
    \centering
    \begin{tabular*}{\textwidth}{@{\extracolsep\fill}lccccccc}
    \hline
    $\delta$ (in \%) & $\leq 10^{-2}$ & $(10^{-2}, 0.1]$ & $(0.1, 1]$ & $(1, 5]$ & $(5, 10]$ & $(10, 20]$ & $(20, 35]$ \\
    \hline
     Normal distributions & 400 & 400 & 400 & 400 & 393 & 404 & 403
     \\
     Gamma distributions & 287 & 505 & 408 & 400 & 420 & 480 & 300 \\
     \hline
    \end{tabular*}
\end{table}

\begin{figure}[!t]%
\centering
\includegraphics[width=.9\textwidth]{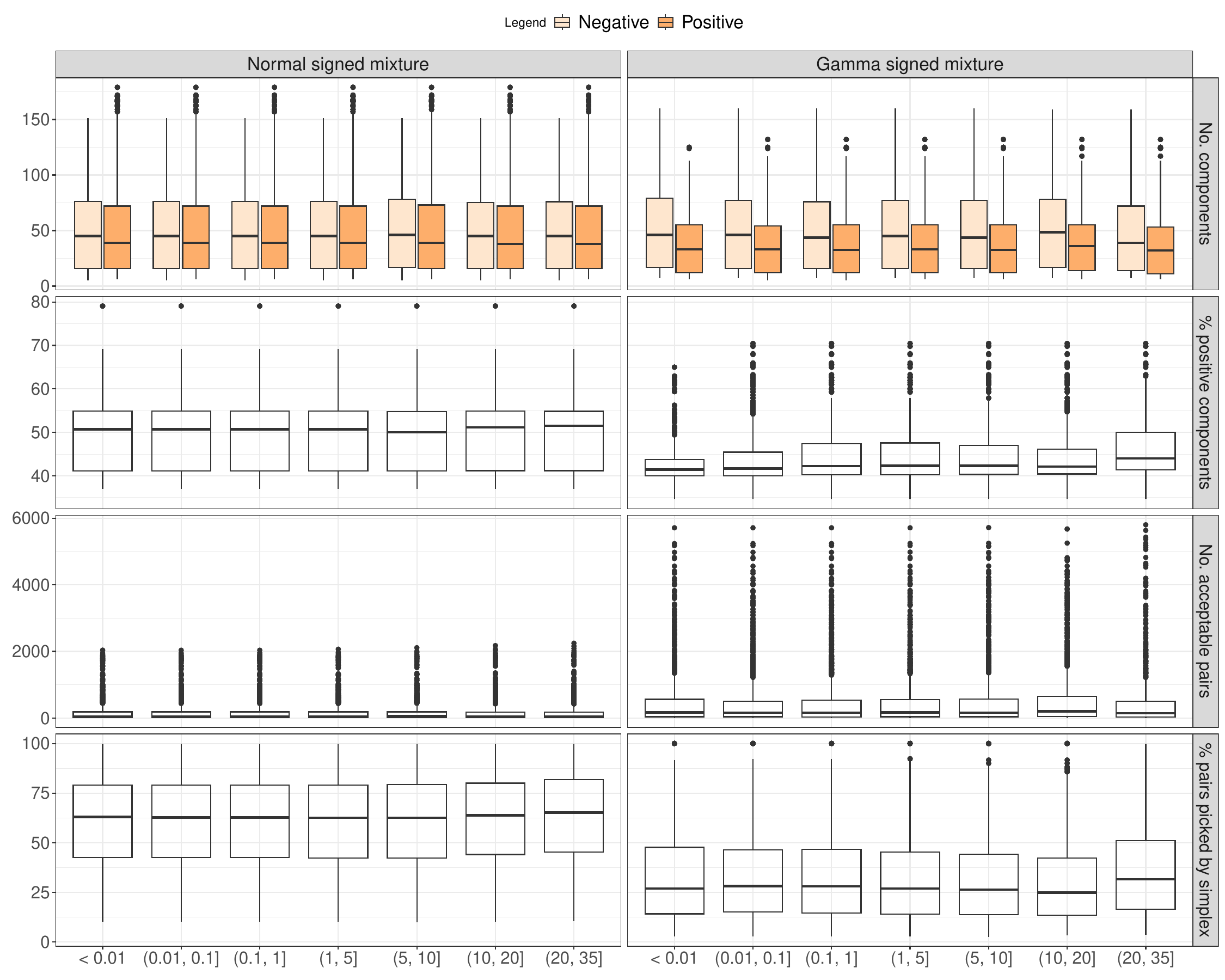}
\caption{Summaries per vanilla average acceptance probability categories (x-axis in \%) of the 2,800 randomly generated signed mixtures of, respectively, Normal distributions (left) and Gamma distributions (right): number of positive and negative weight components (top row), proportion of positive weight components in the model (second row), number of acceptable pairs in the model (third row) and proportion of acceptable pairs selected by the simplex algorithm (bottom row).}
\label{fig:summary-rand-negmix}
\end{figure}

\paragraph*{Comments}
The running time of our method does not depend significantly on the user-specified acceptance probability nor the tolerance level (see Figure \ref{fig:impact-delta-eps}). However, we can point out a consistent pattern regarding the influence of both $\delta$ and $\varepsilon$. Allowing a larger tolerance level leads to a reduced cost since it implies building a partition with less elements. However, opting for a larger acceptance probability happens to increase the running time. In such settings, we end up with
a larger domain to partition and a tolerance level restricted to a smaller range. Hence this results in increasing the number of partition terms, as we aim at a more precise piecewise approximation of the signed mixture. Our method is not designed to efficiently achieve acceptance probability arbitrary close to 1. Instead, users can benefit from reasonably lowering the acceptance probability $\delta$. Obviously, this holds as long as $\delta$ remains larger than the vanilla acceptance probability and the simulation cost does not exceed the advantage of the stratification.

\begin{figure}[!t]
    \centering
    \includegraphics[width = \textwidth]{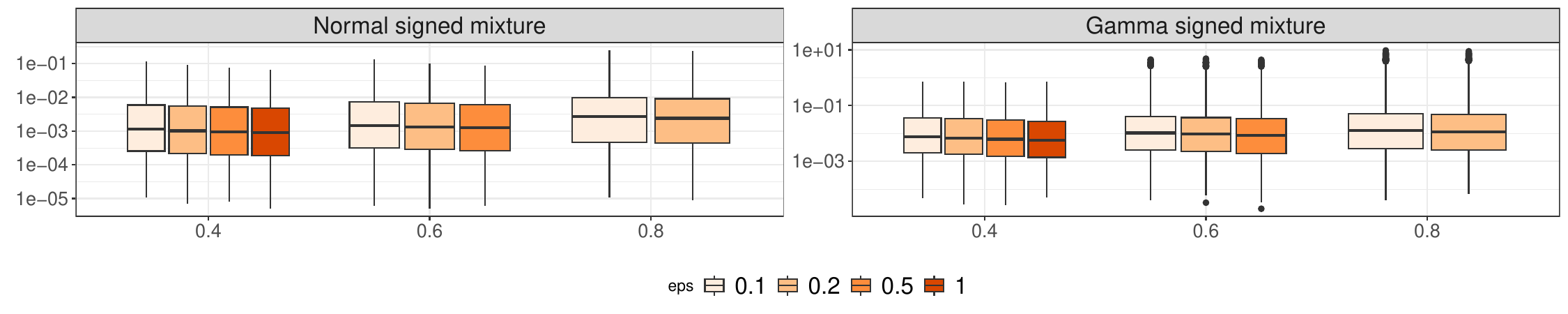}
    \caption{Running time (in sec.) of the stratified method with respect to user-specified acceptance probability $\delta$ (x-axis) and tolerance level $\epsilon$ (\texttt{eps}) for the 2,800 randomly generated signed mixtures of, respectively, Normal distributions (left) and Gamma distributions (right).}
    \label{fig:impact-delta-eps}
\end{figure}

The relative efficiency of our stratified solution compared to the vanilla ranges from around $10^{-5}$ to $10^5$ and unsurprisingly decreases with the vanilla average acceptance probability (see Figure \ref{fig:time-perf}, top row). The stratified approach far outperforms the vanilla method on challenging situations, that is when an accept-reject from the positive part would lead to an average acceptance probability lower than 1\%, a domination found even for very small samples. For a hundred samples, sampling from the positive part of the signed mixture becomes equivalent to, if not better than, the stratified solution when the vanilla average acceptance probability exceeds 5\%. For larger sample sizes, the relative efficiency remains in general larger than $1$. Furthermore, we point out that the median running time of our method for a given sample size is quite stable across the different categories of vanilla acceptance probabilities and mostly lower than the second (see Figure \ref{fig:time-perf}, second row). In comparison, the median running time of the vanilla method strongly depends on its associated acceptance probability (see Figure \ref{fig:time-perf}, third row). This asymmetry means that in situations where the vanilla method performs better, the actual computational benefit is of a negligible scale. Conversely, our method presents a reduction of the simulation cost that is more than substantial in challenging settings, cutting the cost for instance from a few minutes to less than a second.

\begin{figure}[!h]
    \centering
    \includegraphics[width = \textwidth]{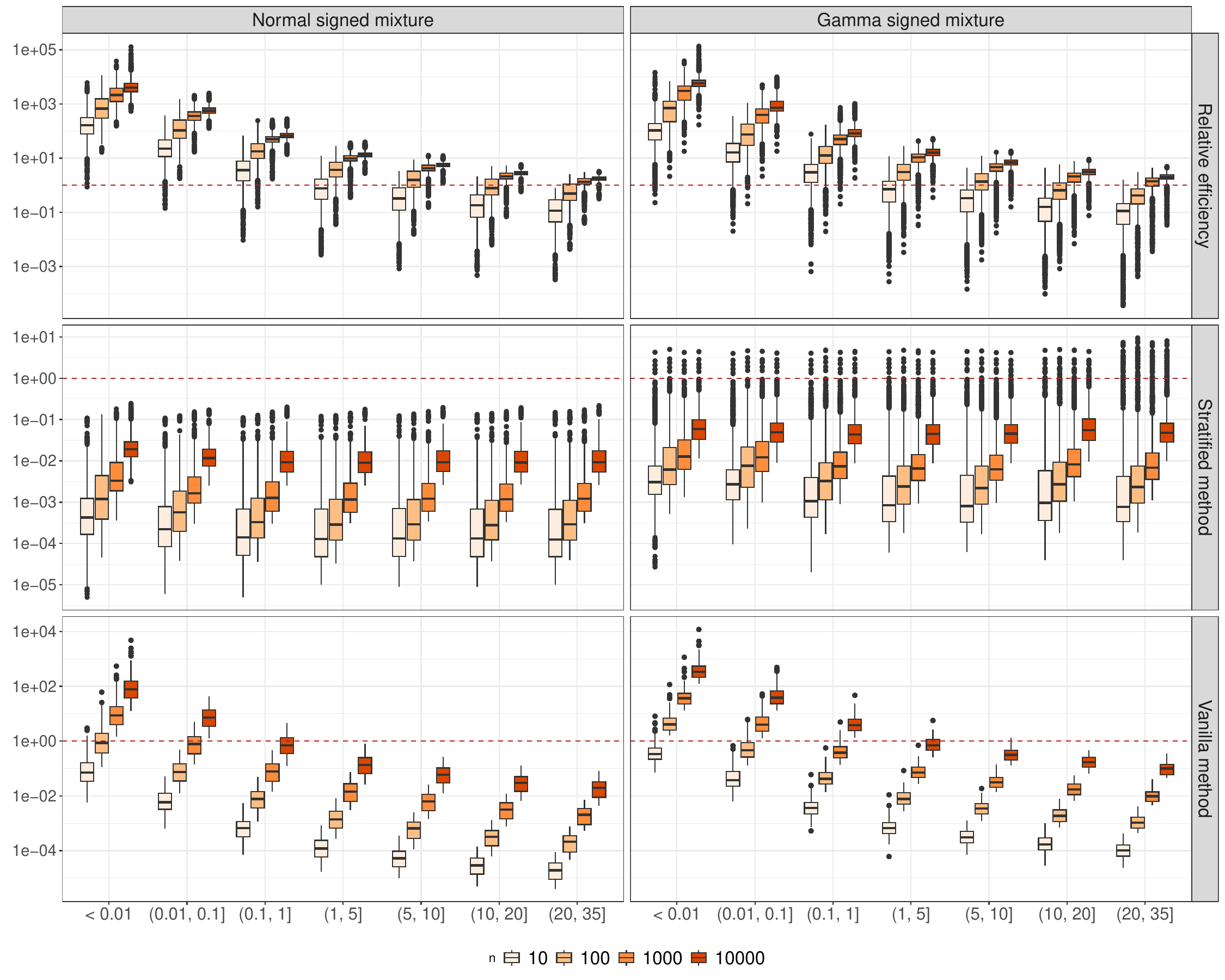}
    \caption{Time performances of accept-reject based methods per vanilla average acceptance probability categories (x-axis in \%) and number of draws $n$ for the 2,800 randomly generated signed mixtures of, respectively, Normal distributions (left) and Gamma distributions (right): 
    relative efficiency of the stratified method compared to the vanilla method (top row), running time (in sec.) of, respectively, the stratified method (second row) and the vanilla method (bottom row).}
    \label{fig:time-perf}
\end{figure}

In the stratified scheme, we have better control of the simulation cost, even in the presence of negative weight residuals (see Figure \ref{fig:rate-strat}), due to the acceptance probability constraint on each pair. This explains the general median stability we observe on Figure \ref{fig:time-perf} regardless of the overall weight of the positive part in the model. The major elements of influence are the computation of the partition and of the pairing using the simplex method. Regarding the partition, we already observed that it does not alter strongly the computational cost of our solution and hence the relative efficiency compared to the vanilla method, but it can be further confirmed with Figure \ref{fig:rel-eff-delta-eps} in Appendix. As for the pairing step, Figure \ref{fig:rel-eff-pairs} illustrates 
the influence of the number $\vert E \vert$ of acceptable pairs on the computational budget. Namely, the cost of our approach increases as the number of pairs increases, and the method becomes less competitive than the vanilla approach.
%that the number $\vert E \vert$ of acceptable pairs has a negative effect in terms of computational budget. 
Indeed, the simplex algorithm is then used to solve an optimization problem involving $2\vert E\vert$ variables and $\vert E\vert + N + P$ constraints. For a moderate number of samples, the efficiency of our solution is reduced when the model contains over a thousand acceptable pairs. In this regime, the simplex may prove more time-consuming than simulating even numerous random variables.

\begin{figure}[!h]
    \centering
    \includegraphics[width = \textwidth]{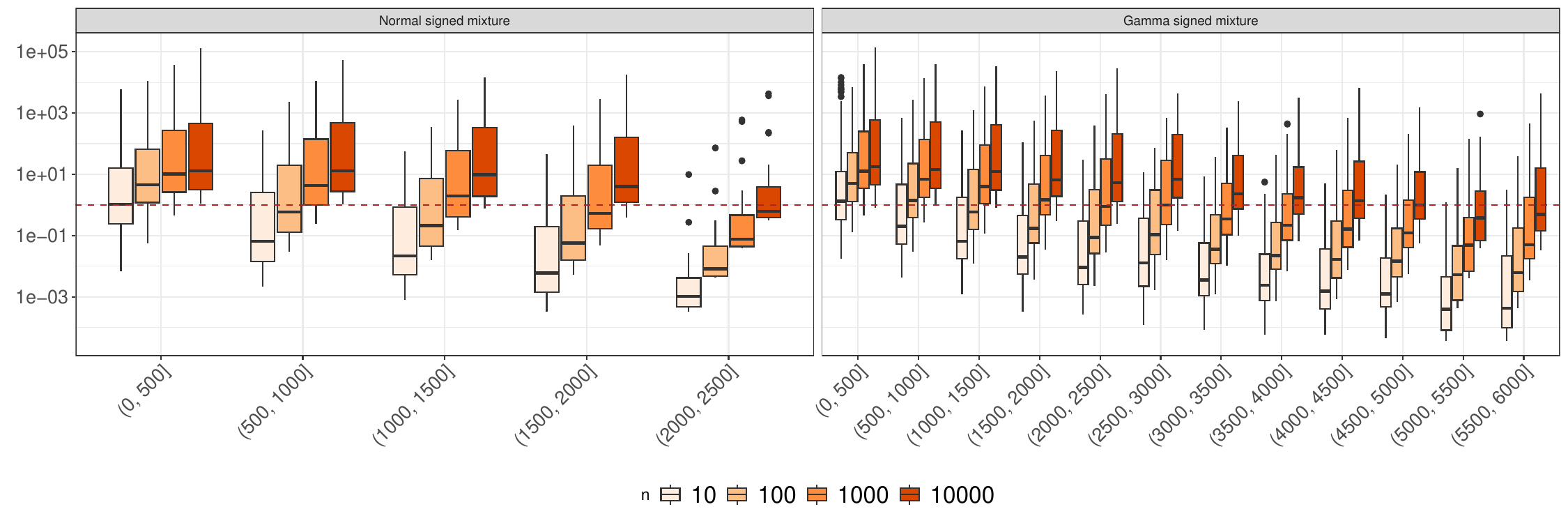}
    \caption{Relative efficiency of the stratified method compared to the vanilla method with respect to the number of acceptable pairs (x-axis) and the number of draws $n$, for the 2,800 randomly generated signed mixtures of, respectively, Normal distributions (left) and Gamma distributions (right).}
    \label{fig:rel-eff-pairs}
\end{figure}

Computing a numerical inverse of the cdf does not exhibit a practical advantage over our accept-reject based method from a computational perspective (see Figure \ref{fig:rel-eff-inv-cdf}). 
Indeed, the median relative efficiency of our method compared to the numerical inverse is close to 1, if not greater. Additionally, the numerical inverse solution only generates samples from an approximate probability measure. 
As shown in the bottom row of Figure \ref{fig:rel-eff-inv-cdf}, this surrogate quantile function is solely beneficial compared to the vanilla method when the latter exhibits low acceptance probability. Yet, our approach is specifically designed to provide an efficient and exact solution in such a setting.

\begin{figure}[!h]
    \centering
    \includegraphics[width = \textwidth]{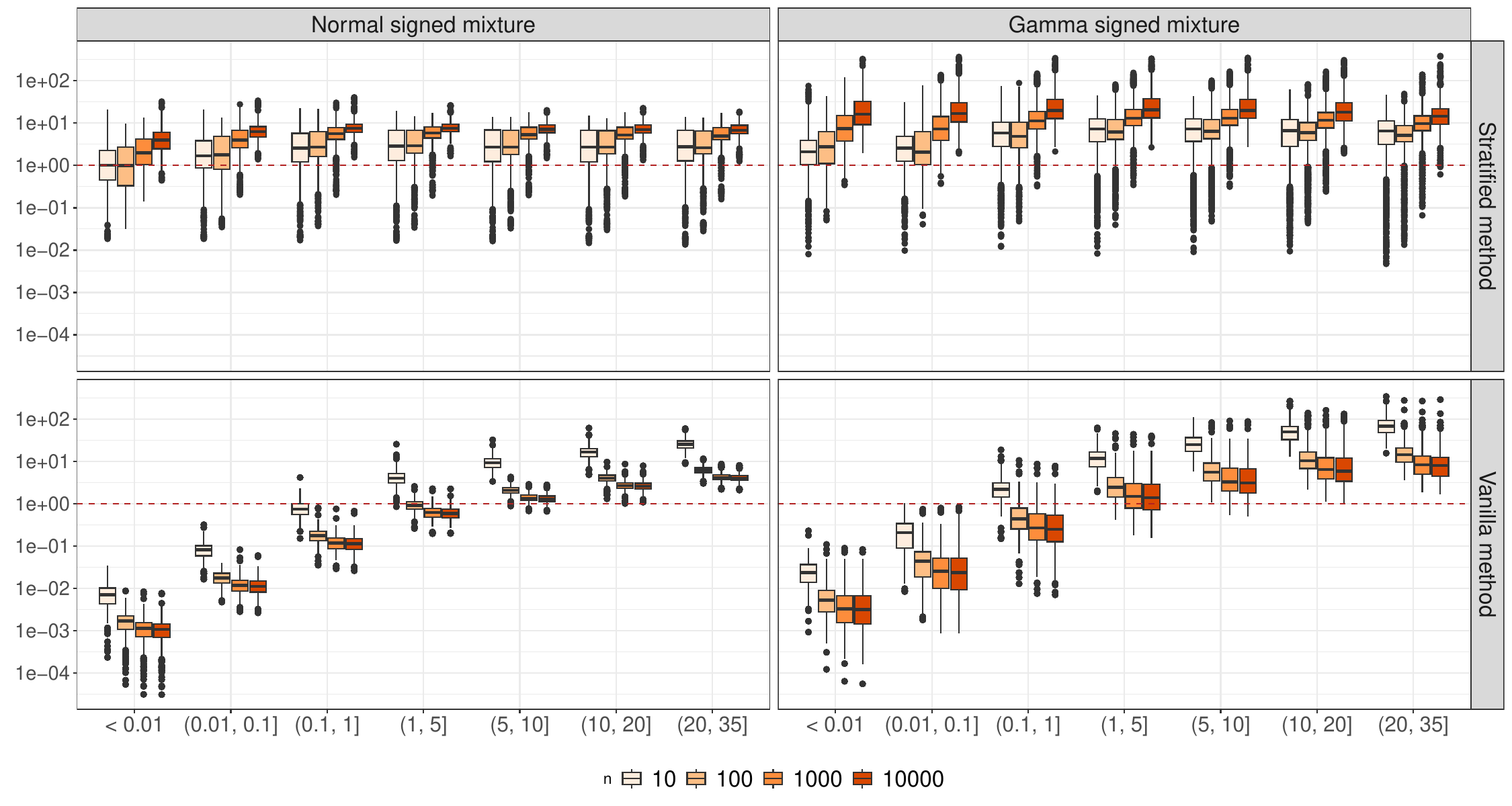}
    \caption{Relative efficiency of, respectively, the stratified method (top row) and the vanilla method (bottom row) compared to the numerical inverse cdf per vanilla average acceptance probability categories (x-axis in \%), and number of draws $n$ for the 2,800 randomly generated signed mixtures of, respectively, Normal distributions (left) and Gamma distributions (right).}
    \label{fig:rel-eff-inv-cdf}
\end{figure}

\section{Conclusions}

The challenge of simulating a signed mixture \eqref{eqn:signed-mixt} surprisingly differs from the standard simulation of an unsigned mixture in that the negative components of \eqref{eqn:signed-mixt} have no natural association with a latent variable. It thus proves impossible to directly eliminate simulations that issue from these negative terms, i.e., to formalize a negative version of accept-reject and one has to resort to more rudimentary approaches. As discussed above, sampling from a signed mixture using only the positive part of the density may prove cumbersome, especially when the weight of the latter is small. While elementary, our alternative approach achieves noticeably superior computational performances by combining a simplex step towards identifying an efficient decomposition of the model into a well-balanced set of two-component mixtures, and a piecewise constant approximation of these two-component distributions. Controlling a lower bound on the average acceptance probability ensures steady performance, regardless of the overall weight of the positive part. Furthermore, this alternative performs most satisfactorily relative to the inverse cdf approach, a feat explained in part by the necessity to numerically invert the cdf, even in cases when the quantile function of both positive and negative components is known.

\section*{Acknowledgements}
A discussion with Murray Pollock (University of Newcastle) was instrumental in sparkling our interest in
the matter. The first author has been partly supported by a senior chair (2016-2021) from
l'Institut Universitaire de France, by a Prairie chair from the Agence
Nationale de la Recherche (ANR-19-P3IA-0001) and by the European Union under the (2023-2030) ERC Synergy grant 101071601 (OCEAN). 
Views and opinions expressed are however those of the author only and do not necessarily reflect those of the European Union or the European Research Council Executive Agency. Neither the European Union nor the granting authority can be held responsible for them.

\section*{Code availability}
The code used throughout the paper is available at \url{https://github.com/jstoehr/negmix}.

\bibliographystyle{abbrvnat}
\bibliography{biblio}

\begin{thebibliography}{18}
\providecommand{\natexlab}[1]{#1}
\providecommand{\url}[1]{\texttt{#1}}
\expandafter\ifx\csname urlstyle\endcsname\relax
  \providecommand{\doi}[1]{doi: #1}\else
  \providecommand{\doi}{doi: \begingroup \urlstyle{rm}\Url}\fi

\bibitem[Beaulieu(1990)]{beaulieu:1990}
N.~Beaulieu.
\newblock An infinite series for the computation of the complementary
  probability distribution function of a sum of independent random variables
  and its application to the sum of {R}ayleigh random variables.
\newblock \emph{IEEE Transactions on Communications}, 38\penalty0 (9):\penalty0
  1463--1474, 1990.
\newblock \doi{10.1109/26.61387}.

\bibitem[Bignami and De~Matteis(1971)]{bignami:dematteis:1971}
A.~Bignami and A.~De~Matteis.
\newblock {A Note on Sampling from Combinations of Distributions}.
\newblock \emph{IMA Journal of Applied Mathematics}, 8\penalty0 (1):\penalty0
  80--81, 1971.

\bibitem[Capp{\'e} et~al.(2008)Capp{\'e}, Douc, Guillin, Marin, and
  Robert]{cappe:douc:guillin:marin:robert:2008}
O.~Capp{\'e}, R.~Douc, A.~Guillin, J.-M. Marin, and C.~P. Robert.
\newblock Adaptive importance sampling in general mixture classes.
\newblock \emph{Statistics and Computing}, 18\penalty0 (4):\penalty0 447--459,
  2008.

\bibitem[Dantzig(1963)]{dantzig:1963}
B.~G. Dantzig.
\newblock \emph{Linear Programming and Extensions}.
\newblock Princeton University Press, Princeton, 1963.

\bibitem[Delaigle and Hall(2010)]{delaigle:hall:2010}
A.~Delaigle and P.~Hall.
\newblock Defining probability density for a distribution of random functions.
\newblock \emph{The Annals of Statistics}, 38\penalty0 (2):\penalty0
  1171--1193, 2010.

\bibitem[Devroye(1985)]{devroye85}
L.~Devroye.
\newblock \emph{{N}on-{U}niform {R}andom {V}ariate {G}eneration}.
\newblock Springer-Verlag, New-York, 1985.

\bibitem[Elston and Glassy(1989)]{elston:glasbey:1989}
D.~A. Elston and C.~A. Glassy.
\newblock Simulating from a mixture of exponential distributions with some
  negatively weighted components.
\newblock \emph{Journal of Statistical Computation and Simulation}, 33\penalty0
  (1):\penalty0 1--9, 1989.

\bibitem[Gumbel(1960)]{gumbel:1960}
E.~J. Gumbel.
\newblock Bivariate exponential distributions.
\newblock \emph{Journal of the American Statistical Association}, 55:\penalty0
  698--707, 1960.

\bibitem[Hubalek and Kuznetsov(2011)]{hubalek:kuznetsov:2011}
F.~Hubalek and A.~Kuznetsov.
\newblock {A convergent series representation for the density of the supremum
  of a stable process}.
\newblock \emph{Electronic Communications in Probability}, 16:\penalty0 84--95,
  2011.
\newblock \doi{10.1214/ECP.v16-1601}.

\bibitem[Kroese et~al.(2019)Kroese, Botev, Taimre, and
  Vaisman]{kroese:etal:2019}
D.~Kroese, Z.~Botev, T.~Taimre, and R.~Vaisman.
\newblock \emph{Data Science and Machine Learning: Mathematical and Statistical
  Methods}.
\newblock Chapman \& Hall/CRC Machine Learning \& Pattern Recognition. CRC
  Press, New York, 2019.
\newblock ISBN 9781000730777.
\newblock URL \url{https://books.google.fr/books?id=F7zADwAAQBAJ}.

\bibitem[Loconte et~al.(2024)Loconte, Sladek, Mengel, Trapp, Solin, Gillis, and
  Vergari]{loconte:etal:2024}
L.~Loconte, A.~M. Sladek, S.~Mengel, M.~Trapp, A.~Solin, N.~Gillis, and
  A.~Vergari.
\newblock Subtractive mixture models via squaring: Representation and learning,
  2024.
\newblock URL \url{https://arxiv.org/abs/2310.00724}.

\bibitem[McLachlan and Peel(2000)]{maclachlan01}
G.~J. McLachlan and D.~Peel.
\newblock \emph{Finite {M}ixture {M}odels}.
\newblock J. Wiley, New York, 2000.

\bibitem[M\"uller et~al.(2012)M\"uller, Ali-L\"oytty, Dashti, Nurminen, and
  Pich\'e]{mueller:2012}
P.~M\"uller, S.~Ali-L\"oytty, M.~Dashti, H.~Nurminen, and R.~Pich\'e.
\newblock Gaussian mixture filter allowing negative weights and its application
  to positioning using signal strength measurements.
\newblock In \emph{2012 9th Workshop on Positioning, Navigation and
  Communication}, pages 71--76, 03 2012.
\newblock ISBN 978-1-4673-1437-4.
\newblock \doi{10.1109/WPNC.2012.6268741}.

\bibitem[Polson and Sokolov(2024)]{polson2024negativeprobability}
N.~Polson and V.~Sokolov.
\newblock Negative probability, 2024.
\newblock URL \url{https://arxiv.org/abs/2405.03043}.

\bibitem[Robert and Casella(2004)]{robert:casella:2004}
C.~Robert and G.~Casella.
\newblock \emph{{M}onte {C}arlo Statistical Methods}.
\newblock Springer, New York, second edition, 2004.

\bibitem[Schuster et~al.(2020)Schuster, Mollenhauer, Klus, and
  Muandet]{schuster:etal:2020}
I.~Schuster, M.~Mollenhauer, S.~Klus, and K.~Muandet.
\newblock Kernel conditional density operators.
\newblock In S.~Chiappa and R.~Calandra, editors, \emph{Proceedings of the
  Twenty Third International Conference on Artificial Intelligence and
  Statistics}, volume 108 of \emph{Proceedings of Machine Learning Research},
  pages 993--1004, Brookline, MA, 2020. PMLR.

\bibitem[Titterington et~al.(1985)Titterington, Smith, and
  Makov]{titterington:smith:makov:1985}
D.~Titterington, A.~Smith, and U.~Makov.
\newblock \emph{{S}tatistical Analysis of Finite Mixture Distributions}.
\newblock wiley, New York, 1985.

\bibitem[Zhang and Zhang(2005)]{zhang:zhang:2005}
B.~Zhang and C.~Zhang.
\newblock Finite mixture models with negative components.
\newblock In P.~Perner and A.~Imiya, editors, \emph{Machine Learning and Data
  Mining in Pattern Recognition}, pages 31--41, Berlin Heidelberg, 2005.
  Springer.
\newblock ISBN 978-3-540-31891-0.

\end{thebibliography}

\begin{appendices}

\section{Two-component signed mixtures}
\label{app:two-comp}

\subsection{Lower bound property}
\begin{lemma}
\label{lemma:a-star}
Assuming two separate probability density functions $f$ and $g$ such that $g$ is absolutely continuous with respect to $f$, then
\begin{align}
a^\star = \sup_{\supp(f)} \frac{g}{f} > 1.
\end{align}
\end{lemma}

\begin{proof}
Let assume $a^\star \leq 1$ and denote $E = \{x \in \supp(f) \mid f(x) = g(x)\}$. We have for all $x \in\supp(f)\setminus E$, $g(x) < f(x)$ and 
\begin{align}
\int_{E} g(x)\mathrm{d} x =  \int_{E} f(x)\mathrm{d} x
\quad\text{and}\quad
\int_{\supp(f)\setminus E} g(x)\mathrm{d} x <  \int_{\supp(f)\setminus E} f(x)\mathrm{d} x.
\end{align}
Since $\supp(g) \subseteq \supp(f)$, we thus have
\begin{align}
1 = \int_{\supp(f)} g(x)\mathrm{d} x < \int_{E} f(x)\mathrm{d} x + \int_{\supp(f)\setminus E} f(x) \mathrm{d} x = 1.
\end{align}
\textit{Reductio ad absurdum} complete.
\end{proof}

\subsection{Results on stratified sampling scheme}
\label{app:strat}

\subsubsection{Average acceptance probability of Algorithm \ref{algo:ar2}}
\label{sec:proof-algo-1}

\paragraph*{Behaviour in the tails} The distribution $m$ restricted to $D_0$ satisfies
\begin{equation*}
\frac{1}{m(D_0)}m(x)\mathds{1}_{D_0}(x) 
\leq \frac{af(D_0)}{(a-1)m(D_0)} \frac{1}{f(D_0)}f(x)\mathds{1}_{D_0}(x).
\end{equation*}
To get one sample from $m$ restricted to $D_0$, we need on average 
\begin{equation*}
M_0 = \frac{af(D_0)}{(a-1)m(D_0)}
\end{equation*}
samples from the distribution $f$ truncated to $D_0$.

\paragraph*{Behaviour in $D_1, \ldots D_{n}$} The distribution $m$ restricted to $D_i$ satisfies
\begin{equation*}
\frac{1}{m(D_i)}m(x)\mathds{1}_{D_i}(x) 
\leq \frac{h_i \vert D_i \vert}{(a-1)m(D_i)} \frac{1}{\vert D_i \vert}\mathds{1}_{D_i}(x).
\end{equation*}
To get one sample from $m$ restricted to $D_i$, we need on average
\begin{equation*}
M_i = \frac{h_i\vert D_i\vert}{(a-1)m(D_i)}
\end{equation*}
samples from the uniform distribution on $D_i$.

\paragraph*{Global behaviour} To get one sample from $m$, we need to propose on average
\begin{equation*}
\sum_{i = 0}^n m(D_i) M_i = \frac{a}{a - 1} f(D_0) + \frac{1}{a-1}\sum_{i = 1}^{n}h_i\vert D_i\vert = M
\end{equation*}
random variables.

\begin{remark}
Sampling from the distribution $f$ restricted to $D_0$ is not necessarily straightforward and might require an accept-reject scheme as well. Both methods based on piecewise proposals have nevertheless still the same acceptance probability on average. If we need $N_0$ samples from a proposal to get one sample from $f$ restricted to $D_0$, Algorithm \ref{algo:ar2} then requires simulating
\begin{equation*}
\tilde{M} = \frac{a}{a - 1} f(D_0)N_0 + \frac{1}{a-1}\sum_{i = 1}^{n}h_i\vert D_i\vert
\end{equation*}
random variables. Conversely, sampling from the dominating piecewise function would require
\begin{equation*}
    M\left\lbrace \frac{af(D_0)}{M(a-1)}N_0 + \frac{1}{M(a - 1)} \sum_{i = 1}^n h_i\vert D_i\vert\right\rbrace = \tilde{M}
\end{equation*}
random variables.
\end{remark}

\subsubsection{Proof of Lemma \ref{lem:accept-proba}}
\label{sec:proof:lem:accept-proba}

\begin{proof}
We have
\begin{align*}
\frac{1}{a-1}\sum_{i = 1}^{n} h_i \vert D_i\vert  \limite{n}{+\infty} \int_{\supp(f)\setminus D_0} m(x) \mathrm{d}x
= 1 - \frac{af(D_0) - g(D_0)}{a-1}.
\end{align*}
Hence, for all $\varepsilon > 0$, there exists $n_\varepsilon$, such that for all $n \geq n_\varepsilon$
\begin{align*}
\left\vert
M - \frac{a - 1 + g(D_0)}{a-1}
\right\vert \leq \varepsilon.
\end{align*}
Given $\varepsilon \in [0, (1 - \delta)/\delta)$, 
if
% \begin{equation*}
% \min\left[
% 1, \frac{(a - 1)\{1 - \delta(\varepsilon + 1)\}}{\delta}
% \right] = 1,
% \end{equation*}
% then
% \begin{equation*}
% \frac{a}{a - 1} + \varepsilon \leq \frac{1}{\delta}
% \quad\text{and}\quad 
% \forall n \geq n_{\varepsilon}, \quad M \leq \frac{a}{a - 1} + \varepsilon.
% \end{equation*}
% Otherwise, we have 
\begin{equation*}
g(D_0) = \frac{(a - 1)\{1 - \delta(\varepsilon + 1)\}}{\delta},
% \quad\text{and}\quad
% \left\vert M - \frac{1}{\delta} + \varepsilon \right\vert \leq \varepsilon.
\end{equation*}
then
\begin{equation*}
\left\vert M - \frac{1}{\delta} + \varepsilon \right\vert \leq \varepsilon.
\end{equation*}
This leads to $1/M \geq \delta$.
%\begin{align*}
%M \leq \varepsilon + 1 + \frac{1 - \delta(\varepsilon + 1)}{\delta} = \frac{1}{\delta}.
%\end{align*}
\end{proof}

\begin{remark}
Under the assumption of Lemma \ref{lem:accept-proba}, we have
\begin{align*}
M  \limite{n}{+\infty} \frac{1}{\delta} - \varepsilon.
\end{align*}
\end{remark}

\begin{remark}
%Assuming $(a - 1)\{1 - \delta(\varepsilon + 1)\} < \delta$, 
A direct consequence of Lemma \ref{lem:accept-proba} is that, if we pick the partition of $\supp(f)\setminus D_0$ such that
\begin{equation*}
    1 - m(D_0) + \varepsilon = \frac{1}{a - 1}\sum_{i = 1}^{n_\varepsilon} h_i\vert D_i\vert,
\end{equation*}
then using the assumption on $g(D_0)$ we get
\begin{equation*}
    M = \frac{a}{a-1}f(D_0) + 1 - m(D_0) + \varepsilon
    = 1 + \frac{1}{a-1}g(D_0) + \varepsilon = \frac{1}{\delta}.
\end{equation*}
\end{remark}

\subsection{Exponential families examples}
\label{app:exp-fam}
Assume that, within the context of Section \ref{sec:2cpt}, the terms $f$ and $g$ are both distributions from the same exponential family 
\begin{equation*}
\mathcal{F} = \left\lbrace
c(\theta)h(x)\exp\{\eta(\theta)^\top T(x)\}\,;\, x\in\mathbb{R}^d,\, \theta\in\Theta  \subseteq \mathbb{R}^q
\right\rbrace.
\end{equation*}
A pairing of $f$ and $g$, parametrized respectively by $\theta^+$ and $\theta^-$, into a two-component signed mixture is thus possible if
\begin{equation}
\label{eqn:a-exp-fam}
a^\star = \sup_{x\in\supp(f)} \{\eta(\theta^-) - \eta(\theta^+)\}^\top T(x) < +\infty.
\end{equation}
% Optima for univariate exponential families are the solution of
% \begin{equation*}
% %\frac{\partial}{\partial x} \frac{g(x\mid\theta^-)}{f(x\mid\theta^+)} = 0
% %\quad \Leftrightarrow \quad
% \sum_{i = 1}^q (\theta_i^- - \theta_i^+) \frac{\partial}{\partial x}T_i(x) = 0.
% \end{equation*}
% A global supremum exists if, and only if, for all $x\in\supp(f)$
% \begin{equation*}
% \sum_{i = 1}^q (\theta_i^- - \theta_i^+) 
% %\left\lbrace\frac{\partial^2}{\partial x^2}T_i(x) + \left(\frac{\partial}{\partial x}T_i(x)\right)^2\right\rbrace 
% \frac{\partial^2}{\partial x^2}T_i(x)
% \leq 0.
% \end{equation*}

\subsubsection{Example of Normal distributions} 
\label{app:normal}
Let $f \equiv \mathcal{N}(\mu^+, \sigma_+^2)$ and $g \equiv \mathcal{N}(\mu^-, \sigma_-^2)$. Since
\begin{equation*}
-\frac{(x - \mu^-)^2}{2\sigma_-^2} + \frac{(x - \mu^+)^2}{2\sigma_+^2} \underset{\pm\infty}{\sim}
-x^2\left(\frac{\sigma_+^2 - \sigma_-^2}{2\sigma_-^2\sigma_+^2}\right)
\end{equation*}
condition \eqref{eqn:a-exp-fam} is fulfilled if $\sigma_-^2 < \sigma_+^2$ (or if $\mu^+ = \mu^-$ and $\sigma_-^2 = \sigma_+^2$ which is of no interest). 
Assuming $\sigma_-^2 < \sigma_+^2$, critical points are then solution of
\begin{equation*}
\frac{\mu^-}{\sigma_-^2} - \frac{\mu^+}{\sigma_+^2} + 2x\left(
-\frac{1}{2\sigma_-^2} + \frac{1}{2\sigma_+^2} 
\right) = 0.
\end{equation*}
We derive a global maximum at
\begin{equation*}
x^\star = \frac{\mu^-\sigma_+^2 - \mu^+\sigma_-^2}{\sigma_+^2 - \sigma_-^2}. 
\end{equation*}
Then
\begin{equation*}
a^\star = \frac{\sigma_+}{\sigma_-} \exp\left\lbrace
\frac{(\mu^+ - \mu^-)^2}{2(\sigma_+^2 -\sigma_-^2)}
\right\rbrace.
\end{equation*}

\paragraph*{Monotonicity of a two-component Normal signed mixture}

Assume $f \equiv \mathcal{N}(0, 1)$ and $g \equiv \mathcal{N}(\mu, \sigma^2)$, with $\mu \geq 0$ and $\sigma < 1$. 
%\begin{lemma}
    The signed mixture $m$ has at most 3 extreme values. More specifically, it admits
    \begin{enumerate}[label = (\roman*)]
        \item a unique global maximum in $(-\infty, 0]$, if 
            \begin{equation*}
                a \geq \sup_{x > \mu} \frac{(x-\mu)g(x)}{\sigma^2 xf(x)}\,;
            \end{equation*}
        \item a local maximum in $(-\infty, 0]$, a local minimum and a local maximum in $[\mu, +\infty)$, otherwise.
    \end{enumerate}
%\end{lemma}

%\begin{proof}
We have for all $x\in\mathbb{R}$
\begin{equation*}
%\frac{\rm d}{\mathrm{d} x} 
m'(x) = \frac{f(x)}{a-1} \left\{\psi(x) -ax\right\},
\quad\text{where}\quad
\psi:x\mapsto \frac{(x-\mu)g(x)}{\sigma^2f(x)}.
%\\
%\frac{\rm d^2}{\mathrm{d} x^2} 
%(a - 1) m''(x) & = a(x^2 - 1)f(x) - \frac{1}{\sigma^2}\left\lbrace\frac{(x-\mu)^2}{\sigma^2} - 1\right\rbrace g(x).
\end{equation*}
The number of solutions to $m'(x) = 0$ then depends on the number of intersection points between $\psi$ and $x\mapsto ax$. The assumption on two-component signed mixtures imposes $g(x)/f(x)\limite{x}{\pm \infty} 0$. Since it happens at exponential speed, we also have $\psi(x)\limite{x}{\pm \infty} 0$. On the other hand, for all $x\in\mathbb{R}$,
\begin{equation*}
    \psi'(x) = \left\{(\sigma^2 - 1)x^2 + x(2\mu - \mu\sigma^2) + \sigma^2 - \mu^2\right\}\frac{g(x)}{\sigma^4f(x)}.
\end{equation*}
A straightforward computation shows that the equation $\psi'(x) = 0$ has two distinct solutions and thus $\psi$ has a global minimum and a global maximum, respectively at
\begin{equation*}
    x_1 = \frac{\mu}{2} +\frac{\mu - \sigma\sqrt{\mu^2\sigma^2+4-4\sigma^2}}{2(1-\sigma^2)}
    \quad\text{and}\quad
    x_2 = \frac{\mu}{2} +\frac{\mu + \sigma\sqrt{\mu^2\sigma^2+4-4\sigma^2}}{2(1-\sigma^2)}.
\end{equation*}
Moreover, since
\begin{equation*}
\psi''(x) = \frac{g(x)}{\sigma^6f(x)} \left\lbrace
(\sigma^2-1)^2 x^3 + Q_2(x)
%\mu x^2\left\lbrace 1 - (\sigma^2-2)^2\right\rbrace + 
%x\left\lbrace 3(\mu^2 - \sigma^2+\sigma^4) -2\mu^2\sigma^2\right\rbrace
%+ 3\mu\sigma^2 - \mu^3-\mu\sigma^4,
\right\rbrace
\end{equation*}
where $Q_2(x)$ is a univariate polynomial of degree 2, $\psi$ changes convexity solely one time in $[x_1, x_2]$. Note that $\psi''(\mu) = 2\mu g(\mu)/\{\sigma^2 f(\mu)\} \geq 0$ and thus the change of convexity happens between $\mu$ and $x_2$. Functions $\psi$ and $x\mapsto ax$ have then at most 3 intersection points. 

\vspace{\baselineskip}
\noindent\textbf{If $\boldsymbol\mu \mathbf{= 0}$,}\hspace{.5em} we have a first obvious solution: $x = 0$. Since $\psi$ is an odd function when $\mu = 0$, the latter solution is unique if
\begin{equation*}
    a \geq \psi'(0) = \frac{1}{\sigma^3} = \sup_{x > 0} \frac{g(x)}{\sigma^2 f(x)}.
\end{equation*}
It is the unique global maximum for $m$, which thus has the same monotonicity as $f$. Otherwise, it is a local minimum and we have two local maxima corresponding to the intersection points solution of
\begin{equation*}
    \exp\left\lbrace\frac{x^2}{2\sigma^2}(\sigma^2-1)\right\rbrace = a\sigma^3,
\end{equation*}
that is $\pm 2\sigma^2\log(a\sigma^3)/(1 - \sigma^2)$.
% \begin{equation*}
%     \frac{2\sigma^2\log(a\sigma^3)}{1 - \sigma^2}
%     \quad\text{and}\quad \frac{2\sigma^2\log(a\sigma^3)}{\sigma^2 - 1}.
% \end{equation*}

\vspace{\baselineskip}
\noindent\textbf{If $\boldsymbol\mu \mathbf{> 0}$,}\hspace{.5em} we do not have a closed form for the critical points. However $\psi$ is a non-positive function on $(-\infty, \mu]$, that is decreasing on $(-\infty, x_1)$ and an increasing convex function on $(x_1, \mu]$. Consequently, there exists a unique intersection point $y^\star_1$ on $(-\infty, 0]$ that corresponds to a local maximum of $m$. The function $x\mapsto ax$ being positive on $(0, \mu]$, if there are two other intersection points, they are necessarily in $(\mu, +\infty)$. If 
\begin{equation*}
    a \geq \sup_{x > \mu} \frac{(x-\mu)g(x)}{\sigma^2 xf(x)},
\end{equation*}
then for all $x > \mu$, $m'(x) < 0$ and as a result $y^\star_1$ is the unique global maximum of $m$. Otherwise, we have two intersection points. The point $y^\star_2$ corresponding to a local minimum of $m$ is bound to be on $(\mu, x_2)$. Nevertheless, note that on $(\mu, x_2)$
\begin{equation*}
\psi(x) -ax \leq (x - \mu)\frac{a^\star}{\sigma^2} -ax = \frac{(a^\star - a\sigma^2)x - a^\star\mu}{\sigma^2}.
\end{equation*}
If $a^\star - a\sigma^2 > 0$, $m$ is decreasing between $\mu$ and $a^\star\mu/(a^\star - a\sigma^2)$ and $y^\star_2 \geq a^\star\mu/(a^\star - a\sigma^2)$.

\smallskip\begin{remark}
    The results for $\mu < 0$ are obtained by symmetry of the problem. Finally the result for the general case of a signed mixture $m$ of $\mathcal{N}(\mu^+, \sigma_+^2)$ and $\mathcal{N}(\mu^-, \sigma_-^2)$ can be derived using that for all $x\in\mathbb{R}$
\begin{equation*}
m(x) \propto a f\left(\frac{x - \mu^+}{\sigma^+}\right) - g\left(\frac{x - \mu^+}{\sigma^+}\right),
\quad\text{with}\quad
\mu = \frac{\mu^- - \mu^+}{\sigma_+} \quad\text{and}\quad \sigma = \frac{\sigma_-}{\sigma_+}.
\end{equation*}
\end{remark}

\subsubsection{Example of Gamma distributions} 
\label{app:gamma}
Let $f \equiv \Gamma(\alpha^+, \beta^+)$, $g \equiv \Gamma(\alpha^-, \beta^-)$ (shape, rate parametrization). 
Condition \eqref{eqn:a-exp-fam} imposes $\alpha^+ \leq \alpha^-$ and $\beta^+ < \beta^-$, so that
\begin{align*}
(\alpha^- - \alpha^+)\log x + (\beta^+ - \beta^-) x & \limite{x}{0^+} -\infty
\quad\text{or}\quad 0,
\\
(\alpha^- - \alpha^+)\log x + (\beta^+ - \beta^-) x & \limite{x}{+\infty} -\infty.
\end{align*}
Assuming this, critical points are solutions of 
\begin{equation*}
\frac{\alpha^- - \alpha^+}{x} + \beta^+ - \beta^- = 0,
\end{equation*}
which leads to a unique global maximum at
\begin{equation*}
x^\star = \frac{\alpha^+ - \alpha^-}{\beta^+ - \beta^-}.
\end{equation*}
Then
\begin{equation*}
a^\star =
\begin{cases}
\displaystyle\frac{\Gamma(\alpha^+)(\beta^-)^{\alpha^-}}{\Gamma(\alpha^-)(\beta^+)^{\alpha^+}}\exp\left\lbrace (\alpha^+ - \alpha^-) (1 - \log x^\star)\right\rbrace
& \quad\text{if $\alpha^+ < \alpha^-$,}
\\
\vspace{.25\baselineskip}
\\
\displaystyle\left(\frac{\beta^-}{\beta^+}\right)^{\alpha^+}
& \quad\text{if $\alpha^+ = \alpha^-$.}
\end{cases}
\end{equation*}
%and
%\begin{equation*}
%a^\star = \left(\frac{\beta^-}{\beta^+}\right)^{\alpha^+}
%\quad\text{if $\alpha^+ = \alpha^-$}.
%\end{equation*}

\paragraph*{Monotonicity of a two-component Gamma signed mixture}

The arguments for studying the monotonicity are similar to those used for the Gaussian case. For all $x > 0$,
\begin{equation*}
m'(x) = \frac{f(x)}{(a-1)x} \left\{\psi(x) - a(\beta^+ x + 1 - \alpha^+)\right\},
\quad\text{where}\quad
\psi:x\mapsto \frac{(\beta^- x - \alpha^- + 1)g(x)}{f(x)}.
\end{equation*}

\noindent\textbf{If $\boldsymbol{\alpha^+ = \alpha^-}$,}\hspace{.5em} first and second derivatives of $\psi$ write as
\begin{align*}
\psi'(x) & = \left\{
\beta^-(\beta^+ - \beta^-)x + \beta^+ + \alpha^-(\beta^- -\beta^+)
\right\}
\frac{g(x)}{f(x)} \\
\psi''(x) & = \left[
\beta^-(\beta^- - \beta^+)x - \{\beta^+ +\beta^- + \alpha^-(\beta^- -\beta^+)\}
\right]\frac{(\beta^- -\beta^+)g(x)}{f(x)}.
\end{align*}
We thus have a unique global maximum and a single change of convexity. 
\begin{itemize}[leftmargin=*]
    \item If $\alpha^- = \alpha^+ > 1$, then $a(1 - \alpha^+) < \psi(0)$ and $m$ admits a unique global maximum on $(0, + \infty)$.
    \item If $\alpha^- = \alpha^+ \leq 1$, $a(1 - \alpha^+) \geq \psi(0)$ and $m$ admits a local minimum and local maximum on $(0, + \infty)$ solely when
    \begin{equation*}
        a < \sup_{x \geq 0} \frac{\beta^-(\beta^- x + 1 - \alpha^-)}{\beta^+(\beta^+ x + 1 - \alpha^-)} \exp\left\{(\beta^+ - \beta^-)x\right\}.
    \end{equation*}
    Otherwise, $m$ is decreasing on $(0, + \infty)$.
\end{itemize}

\vspace{\baselineskip}
\noindent\textbf{If $\boldsymbol{\alpha^+ < \alpha^-}$,}\hspace{.5em} 
\begin{align*}
\psi'(x) = \Big[
\beta^-(\beta^+ - \beta^-)x^2 & + \{\beta^+(1-\alpha^-) + \beta^-(2\alpha^- - \alpha^+)\}x
\\ &  + (\alpha^- -\alpha^+)(1-\alpha^-)
\Big]\frac{g(x)}{xf(x)}
\end{align*}
The univariate polynomial has necessarily two real roots $x_1$ and $x_2$ (otherwise $\psi$ would be a continuous decreasing function on $[0, +\infty)$ and hence constant since its limit at 0 and $+\infty$ is 0). It is straightforward to show that the smallest root is non-positive when $\alpha^- \leq 1$ and non-negative when $\alpha^- > 1$ while the largest is always positive. The convex properties are identical to the Gaussian example as 
\begin{equation*}
    \psi''(x) = \frac{g(x)}{x^2f(x)}\left\{
    \beta^-(\beta^- - \beta^+)^2 x^3 + Q_2(x)
\right\},
\end{equation*}
where $Q_2(x)$ is a univariate polynomial of degree 2.
\begin{itemize}[leftmargin=*]
    \item If $\alpha^+ < 1$, then $a(1 - \alpha^+) > 0$. $m$ admits a local minimum and local maximum on $(\max\{0, (\alpha^- - 1)/\beta^-\}, + \infty)$ solely when
    \begin{equation*}
        a < \sup_{x \geq 0} \frac{\{\beta^- x + 1 - \alpha^-\}g(x)}{\{\beta^+ x + 1 - \alpha^+\}f(x)}.
    \end{equation*}
    Otherwise, $m$ is decreasing on $(0, + \infty)$.
    \item If $\alpha^+ \geq 1$,  then $a(1 - \alpha^+) \leq 0$. The behaviour depends on the relative position of the modes of each component.
    \begin{itemize}[label = \textbullet]
        \item If $\beta^-(\alpha^+ -1) < \beta^+(\alpha^- -1)$, then $m$ admits a local maximum in $[0, (\alpha^+ -1)/\beta^+]$. It is then
        decreasing on $[(\alpha^+ -1)/\beta^+, + \infty)$ when
        \begin{equation*}
            a \geq \sup_{\beta^- x + 1 - \alpha^-> 0} \frac{\{\beta^- x + 1 - \alpha^-\}g(x)}{\{\beta^+ x + 1 - \alpha^+\}f(x)}.
        \end{equation*}
        Otherwise, $m$ admits a local minimum and a local maximum within the latter interval.
        \item If $\beta^-(\alpha^+ -1) > \beta^+(\alpha^- -1)$, then $m$ admits a local maximum in $[(\alpha^+ -1)/\beta^+, +\infty)$. On $[0, (\alpha^+ -1)/\beta^+]$, $m$ is increasing when
        \begin{equation*}
            a \geq \sup_{\beta^- x + 1 - \alpha^- < 0} \frac{\{\beta^- x + 1 - \alpha^-\}g(x)}{\{\beta^+ x + 1 - \alpha^+\}f(x)}.
        \end{equation*}
        Otherwise, $m$ admits a local maximum and a local minimum within the latter interval.
        \item If $\beta^-(\alpha^+ -1) = \beta^+(\alpha^- -1)$, both components have the same mode that is the unique global maximum of $m$ when
        \begin{equation*}
            a \geq \sup_{\beta^+ x + 1 - \alpha^+ \neq 0} \frac{\{\beta^- x + 1 - \alpha^-\}g(x)}{\{\beta^+ x + 1 - \alpha^+\}f(x)} = \psi'\left(\frac{\alpha^+ - 1}{\beta^+}\right),
        \end{equation*}
        and a local minimum otherwise. In the latter situation, $m$ admits two local maxima, one in $(0, (\alpha^+ -1)/\beta^+)$ and one in $((\alpha^+ -1)/\beta^+, +\infty)$.
    \end{itemize}
\end{itemize}

\subsubsection{Construction of the partition}
\label{app:hist}

% Let consider a target acceptance probability $\delta$ and a tolerance level $\varepsilon$. In what follows, $\mu^+$ and $\mu^-$ denote
% \begin{itemize}
%     \item either the means of the positive and negative weight Normal components, respectively;
%     \item or the modes $\max\{0, (\alpha^{+} - 1)/\beta^+\}$ and $\max\{0, (\alpha^{-} - 1)/\beta^-\}$ of the positive and negative weight Gamma components, respectively.
% \end{itemize}

We compute $D_0 = (q_\alpha, q_{1 - \alpha})^c$, where $q_\alpha$ and $q_{1 - \alpha}$ are respectively  $\alpha$ and $1-\alpha$-quantiles of $g$, with
\begin{equation*}
\alpha = \frac{(a - 1)\{1 - \delta(\varepsilon + 1)\}}{2\delta}.
\end{equation*}
In the specific setting of a two-component Gamma signed mixture with both shape parameters larger than 1, we consider $D_0 = [q_{1 - 2\alpha}, + \infty)$.

We partition $D_0^c$ into $S$ subsets $D_1, \ldots, D_S$ relying on the monotonic properties of the signed mixture. The aim is to decide whether, on subdivisions $[x_i, x_{i+1}[$ of $D_i$,  we use 
    \begin{equation*}
    (\mathrm{A})\quad h_i = \sup_{[x_i, x_{i+1}[} (af - g)(x)
    \quad\text{or}\quad
    (\mathrm{B})\quad h_i = \sup_{[x_i, x_{i+1}[} af(x) - \inf_{[x_i, x_{i+1}[} g(x).
    \end{equation*}
On each subset, the signed mixture has one of the following properties:
\begin{enumerate}
    \item the signed mixture is a monotonic function. On such a subset, we use the version (A) on every subdivision $[x_i, x_{i+1}[$;
    \item the signed mixture changes monotonicity only once on the subset. For all subdivisions $[x_i, x_{i+1}[$ such that $m'(x_i)m'(x_{i+1}) > 0$, we use the version (A). Otherwise, we use the version (B) but that happens solely once;
    \item the signed mixture changes monotonicity more than once on the interval. On such subset, we use the version (B) on every subdivision $[x_i, x_{i+1}[$.
\end{enumerate}
Note that for two-component Gamma and Normal signed mixtures, we can restrict ourselves to use only the first two types of subsets by numerically computing some of the local extrema.

For a given subset $D_s$, $1\leq s \leq S$, we start with the partition $[x_1, x_{2}[, \ldots, [x_{n-1}, x_{n}[$, such that $x_{i+1} - x_i = \vert D_0 \vert /100$, $1\leq i \leq n$. The length of each partition element of $D_s$ is divided by two until we achieve
\begin{equation*}
    \sum_{x_i \in D_s} (x_{i + 1} - x_i)h_i = m(D_i) + \frac{\varepsilon}{S + 1}.
\end{equation*}

\section{Pairing mechanism}
\label{app:pairing}

\subsection{Proof of Lemma \ref{lem:opt-pair}}
\label{sec:proof:lem:opt-pair}

\begin{proof}
As mentioned in the paper, to get one sample from $m$, we need to propose in average $C$ samples from $\pi$. Let now detail the number $N$ of proposed sample from $\pi$ according to the sampling strategy of Lemma \ref{lem:opt-pair}.
The probability to randomly a pick pair $(i,j)\in F$ is $(\omega_{ij}^+ - \omega_{ij}^-)/C$, while the one to pick a residual $i\in\{1, \ldots, P\}$ is $r_i/C$.
To get one sample from a pair $(i,j)\in F$, the vanilla scheme requires proposing $\omega_{ij}^+/(\omega_{ij}^+ - \omega_{ij}^-)$ random variables, while, by assumption, the piecewise sampling scheme requires less than $1/\delta$ random variables. When sampling from a residual, we have an exact and immediate sampler that requires solely to propose one draw. Overall, we then have
\begin{equation*}
N \leq \sum_{(i,j)\in F} \frac{\omega_{ij}^+}{C}
\mathds{1}_{\{(1 - \delta)\omega_{ij}^+ - \omega_{ij}^- \geq 0\}}
+ \sum_{(i,j)\in F} \frac{\omega_{ij}^+ - \omega_{ij}^-}{\delta C}
\mathds{1}_{\{(1 - \delta)\omega_{ij}^+ - \omega_{ij}^- < 0\}}
+ \sum_{i = 1}^P \frac{r_i}{C}.
\end{equation*}
Since 
\begin{equation*}
\sum_{(i,j) \in F} \omega_{ij}^+ + \sum_{i = 1}^P r_i = \sum_{i = 1}^P \omega_i^+,
\end{equation*}
we end up with
\begin{align*}
C \times N 
& \leq
\sum_{i = 1}^P \omega_i^+
- \sum_{(i,j)\in F} \omega_{ij}^+
\mathds{1}_{\{(1 - \delta)\omega_{ij}^+ - \omega_{ij}^- < 0\}}
+ \sum_{(i,j)\in F} \frac{\omega_{ij}^+ - \omega_{ij}^-}{\delta}
\mathds{1}_{\{(1 - \delta)\omega_{ij}^+ - \omega_{ij}^- < 0\}}
\\
& \leq
\sum_{i = 1}^P \omega_i^+ 
+ \frac{1}{\delta} \sum_{(i,j)\in F} \left\lbrace (1 - \delta)\omega_{ij}^+ - \omega_{ij}^- \right\rbrace \mathds{1}_{\{(1 - \delta)\omega_{ij}^+ - \omega_{ij}^- < 0\}}.
\end{align*}
\end{proof}

\subsection{Objective function for the simplex method}
\label{sec:app:simplex}

Minimizing the objective function \eqref{eqn:obj-simplex} provides the optimal pairing for Lemma \ref{lem:opt-pair}. Let $\{\tilde\omega_{ij}^+, \tilde\omega_{ij}^-\}_{(i,j)\in E}$ be a minimizer of \eqref{eqn:obj-simplex}, which, as a reminder, is given by
\begin{equation*}
\sum_{(i,j)\in E} \left\lbrace (1 - \delta)\omega_{ij}^+ - \omega_{ij}^- \right\rbrace,
\end{equation*}
and assume there exists a pair $(k, \ell)\in E$ such that $(1 - \delta)\tilde\omega_{k\ell}^+ - \tilde\omega_{k\ell}^- > 0$. Then,
\begin{align*}
\sum_{(i,j)\in E} \left\lbrace (1 - \delta)\tilde\omega_{ij}^+ - \tilde\omega_{ij}^- \right\rbrace
& > \sum_{(i,j)\in E\setminus\{(k, \ell)\}} \left\lbrace (1 - \delta)\tilde\omega_{ij}^+ - \tilde\omega_{ij}^- \right\rbrace
\\
& \hspace{2em}+ \left\lbrace (1 - \delta)\omega_{k\ell}^+ - \omega_{k\ell}^- \right\rbrace\mathds{1}_{\{(\omega_{k\ell}^+, \omega_{k\ell}^-) = (0,0)\}}.
\end{align*}
This contradicts the fact that $\{\tilde\omega_{ij}^+, \tilde\omega_{ij}^-\}_{(i,j)\in E}$ is a minimizer of \eqref{eqn:obj-simplex}. 
Therefore a minimizer $\{\tilde\omega_{ij}^+, \tilde\omega_{ij}^-\}_{(i,j)\in E}$ of \eqref{eqn:obj-simplex} satisfies for all $(i,j)\in E$, $(1 - \delta)\tilde\omega_{ij}^+ - \tilde\omega_{ij}^- \leq 0$.
Consequently, \eqref{eqn:obj-simplex} has the same set of minimizers than
\begin{equation*}
\sum_{(i,j)\in E} \left\lbrace (1 - \delta)\omega_{ij}^+ - \omega_{ij}^- \right\rbrace 
\mathds{1}_{\left\lbrace (1 - \delta)\omega_{ij}^+ - \omega_{ij}^-  < 0 \right\rbrace}.
\end{equation*}

\section{Numerical inversion of the cdf}
\label{app:inv-cdf}

Consider $n$ ordered points $q_1, \ldots, q_n$ in the support of $f$ and $p_1, \ldots, p_n$ the value of the cdf associated with $m$ at these points, that is $p_i = m((-\infty, q_i])$, $1\leq i\leq n$. Furthermore, set a user-specified precision $\varepsilon$. In the paper, we used $\varepsilon = 10^{-10}$.

The set of points $q_1, \ldots, q_n$ and $p_1, \ldots, p_n$ provides a piecewise affine approximation of the inverse cdf. The aim of our numerical inversion is to refine the affine approximation so that we can find the quantile associated with a probability arbitrary close to a point $u\in[0, 1]$ using one of the following steps.

\paragraph*{Step A} Assume we have $u\in[p_i, p_{i +1}]$. We compute the preimage $q^\star$ of $u$ by the affine transformation on $[p_i, p_{i +1}]$
\begin{equation*}
    x \mapsto \frac{p_{i+1} - p_i}{q_{i+1} - q_i} x + \frac{p_{i}q_{i+1} - p_{i+1}q_i}{q_{i+1} - q_i},
\end{equation*}
that is
\begin{equation*}
    q^\star = \frac{(q_{i+1} - q_i)u - p_{i}q_{i+1} + p_{i+1}q_i}{p_{i+1} - p_i}.
\end{equation*}
Then we compute the cdf at $q^\star$ and denote $p^\star$ its value. This yields a new interval containing $u$ that is strictly included in $[p_i, p_{i +1}]$. We now apply the same procedure on that interval. We repeat the process until we get a value $p^\star$ such that $\vert u - p^\star\vert < \varepsilon$.

\paragraph*{Step B} If we deal with a distribution that has an unbounded support, tails should be treated separately. Assume we have $u < p_1$. We use a scheme similar to the above except we take the preimage by the affine transformation based on the two first points $(p_1, p_2)$ larger than $u$. Here we stop when we find a point $p^\star \leq u + \varepsilon$. If $u > p_n$, the reasoning is the same except we use the last two points smaller than $u$ and we stop when $p^\star \geq u - \varepsilon$. Now, either this ending point satisfies $\vert u - p^\star\vert < \varepsilon$, or we apply Step A starting with the interval $[p^\star, p_{1}]$ for left tail or $[p_n, p^\star]$ for right tail.

\section{Random generator of signed mixture models}
\label{app:rand-mix}

We used two different methods to generate the benchmark models. Both methods start by randomly setting an initial number $K$ of positive weight components in the model. The number $K$ is drawn uniformly between $k_{\min}$ and $k_{\max}$ for the following sets $\{5, \ldots, 10\}$, $\{10, \ldots, 30\}$, $\{30, \ldots, 50\}$, and $\{50, \ldots, 100\}$. Once the number of positive weight components is set, we randomly draw the associated parameter values.
\begin{itemize}
    \item For Normal signed mixtures, the mean $\mu^+$ is drawn uniformly in $[0, 20]$ and the standard deviation is drawn according to a Gamma distribution $\Gamma(3, 2.5)$ (shape, rate parametrization).
    \item For Gamma signed mixtures, the shape parameter $\alpha^+$ is drawn according to a Gamma distribution $\Gamma(4, 0.5)$ and the rate parameter $\beta^+$ to a Gamma distribution $\Gamma)(2, 0.7)$.
\end{itemize}
We then randomly set the number of negative weight components (1 or 2) that are initially related to each positive weight component. The parameter value for the negative components as well as the weights are then computed to ensure a benchmark model for which the vanilla average acceptance probability $p$ ranges in $[p_{\min}, p_{\max}]$ for the following sets $[0, 10^{-4}]$, $(10^{-4}, 0.001]$, $(0.001, 0.01]$, $(0.01, 0.05]$, $(0.05, 0.1]$, $(0.1, 0.2]$ and $(0.2, 0.3]$. For each set of values for $K$ and $p$, and for each method, we generated 50 benchmarks.

\subsection{First method}

The first method is based on the properties of two-component signed mixtures. For a given positive weight component, we compute the parameter value of the associated negative weight component such that $a^\star \in [1/(1-p_{\max}), 1/(1-p_{\min})]$. The weights for this two-signed component are the ones associated with $a^\star$. If the positive weight component is associated with more than one negative component, we repeat this procedure for each negative component. As a result, we thus obtain a collection of two-component signed mixtures that all have the targeted acceptance probability. A convex combination with uniform rates of these mixtures yields a signed mixture with the vanilla average acceptance probability we aim at.

If the overall acceptance probability is lower than $p_{\max}$, we randomly decide to add positive weight components that can either balance some of the negative components already included or that can balance none. We can easily determine the maximal weight to assign to such single components so the acceptance probability remains lower than $p_{\max}$. Indeed, assume we add $\tilde{K}$ positive weight components. The new normalized signed mixture writes as
\begin{equation*}
    \frac{
\sum_{i = 1}^K \omega_i^+f_i - \sum_{j = 1}^N \omega_j^- g_j + \sum_{k = 1}^{\tilde{K}} r_k f_{K + k}
}{
1 + \sum_{k = 1}^{\tilde{K}} r_k
}.
\end{equation*}
The latter is associated with a vanilla acceptance probability lower than $p_{\max}$ as long as
\begin{equation*}
\sum_{k = 1}^{\tilde{K}} r_{k} \leq \frac{p_{\max} \sum_{i = 1}^K \omega_i^+ - 1}{1 - p_{\max}}.
\end{equation*}

In a given benchmark, a negative weight component is hence not naturally paired with a single positive weight component. This method aims at providing benchmarks such that the number of acceptable pairs in the model is quite important. They constitute a good basis to challenge the performances of the simplex method as it has to narrow down the pairs involved in a pairing from a large number of initial acceptable pairs.

\subsection{Second method}
After generating all the positive weight components, accounting for multiplicity when more than one negative weight component is associated, for a given positive weight component, we randomly draw the parameter value of the associated negative weight component such that $a^\star \leq 10$. 
This constraint ensures the two-component signed mixture does not have a vanilla acceptance probability larger than 0.90 in the worst case scenario, making the next step easier. As opposed to the previous method, we now consider the linear combination $af - g$ with $a$ drawn uniformly in $[0, a^\star]$. The resulting function takes negative values on the support of $f$ and, hence, does not define a distribution anymore.

We use all the positive components $f_i$, $1\leq i \leq K-1$ generated, except $f$, to balance the negative part of that function. First, we make sure that all positive components together have enough mass over the set of negative values, that is the function is not negative in the tails of all the possible positive components. When necessary we add one or more positive weight components (we still denote $K$ the overall number of positive weight components). We then compute the weights $\tilde\omega_i^+$ such that
\begin{equation*}
    af - g + \sum_{i = 1}^{K - 1} \tilde\omega_i^+ f_i \geq 0.
\end{equation*}
That yields a collection of $K$ signed mixtures, each one having solely one negative component and associated with a vanilla acceptance probability $p_i$. We consider a convex combination of these signed mixtures to control the acceptance probability associated with the vanilla method and set it to $\min_i p_i$. This is usually not enough to ensure that $p$ ranges in $[p_{\min}, p_{\max}]$. However, we can easily modify the model to satisfy this constraint by adding a two-signed component mixture to the model. We select at random a positive weight component included in the model and we build from it a two-signed mixture $a^\star f - g$ that fulfills the constraint on $p$. The new normalized signed mixture writes as
\begin{equation*}
    \frac{
\sum_{i = 1}^K \omega_i^+f_i - \sum_{j = 1}^N \omega_j^- g_j + \lambda(a^\star f - g)
}{
1 + \lambda(a^\star -1)
}.
\end{equation*}
The latter is associated with a vanilla acceptance probability lower than $p_{\max}$ as long as
\begin{equation*}
\lambda \leq \frac{p_{\max} \sum_{i = 1}^K \omega_i^+ - 1}{a^\star(1 - p_{\max}) - 1}.
\end{equation*}
This method aims at providing benchmarks that exhibit negative weight residuals. Such benchmarks allow to study the performances of the stratified method when the residual mixture obtained after the pairing step degrades the acceptance probability of the procedure (see Figure \ref{fig:rate-strat}).

\section{Supplementary material on methods comparison}

\begin{figure}[h!]
    \centering
    \includegraphics[width = \textwidth]{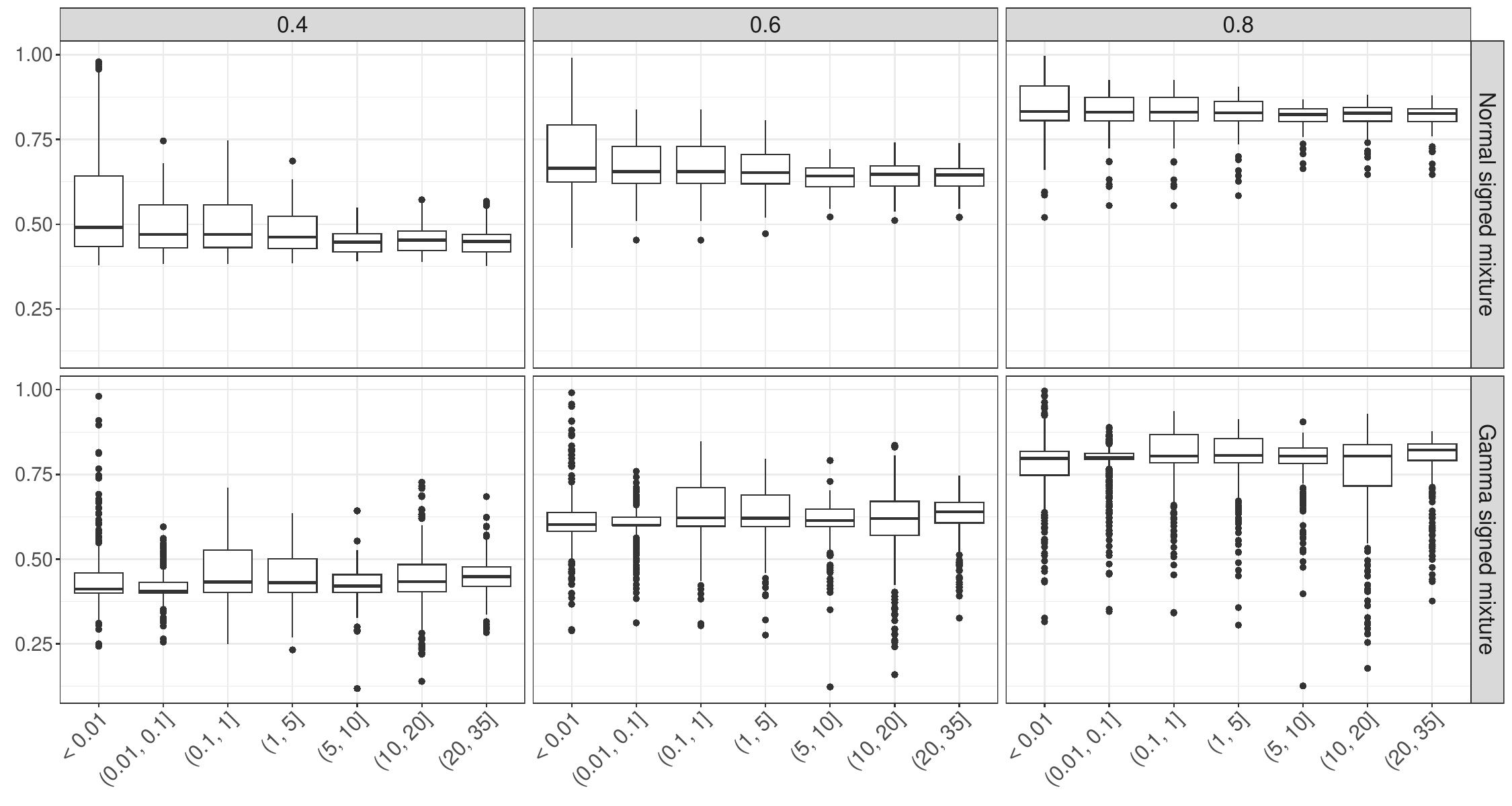}
    \caption{Theoretical acceptance probability of the stratified sampling scheme with respect to the vanilla average acceptance probability categories (x-axis in \%) and user-specified acceptance probability $\delta$ for the 2,800 randomly generated signed mixtures of, respectively, Normal distributions (top row) and Gamma distributions (bottom row). An acceptance probability lower than $\delta$ signals the presence of negative weight residuals.}
    \label{fig:rate-strat}
\end{figure}

\begin{figure}[h!]
    \centering
    \includegraphics[width = \textwidth]{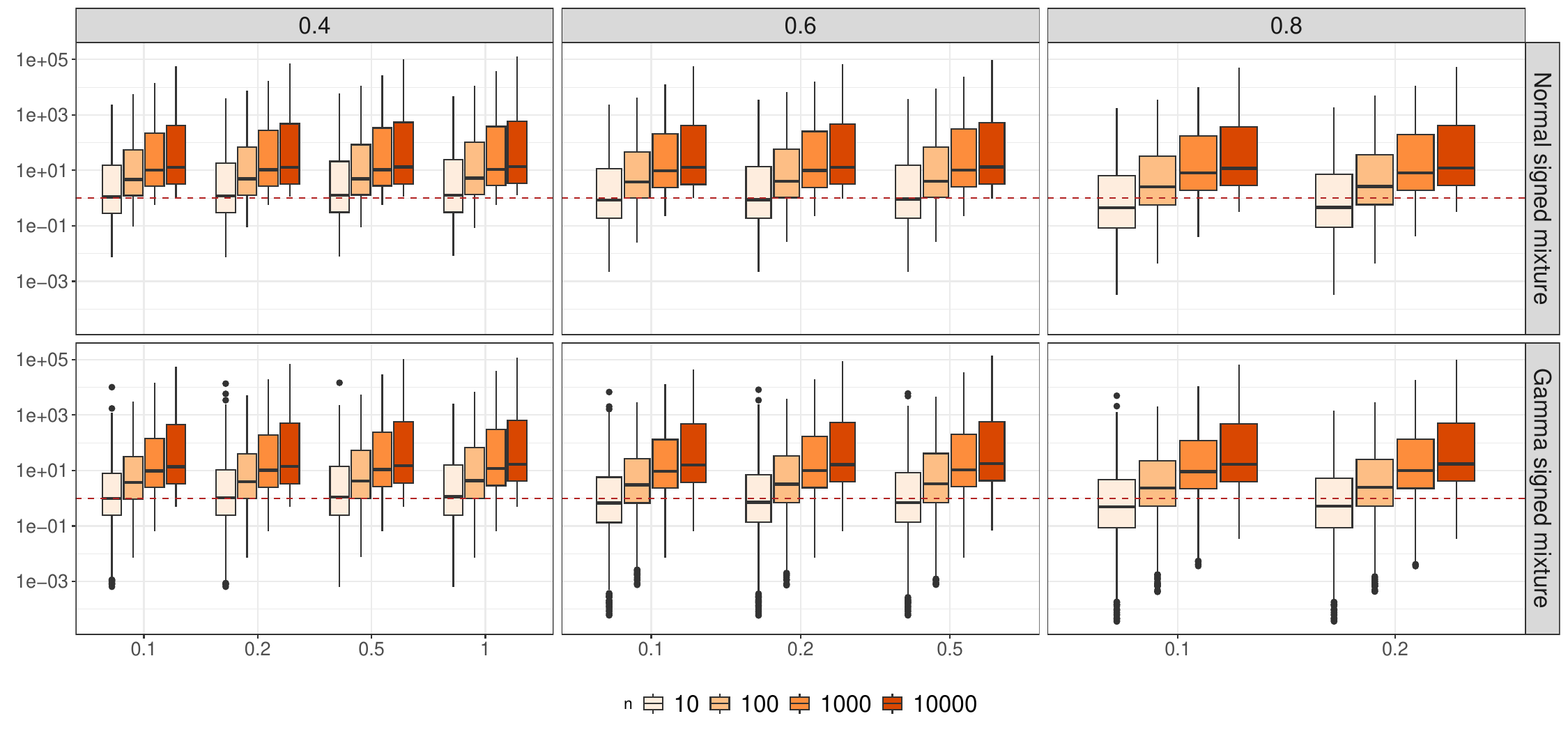}
    \caption{Relative efficiency of the vanilla method compared to the stratified method with respect to user-specified acceptance probability $\delta$, tolerance level $\epsilon$ (x-axis) and number of draws $n$, for the 2,800 randomly generated signed mixtures of, respectively, Normal distributions (top row) and Gamma distributions (bottom row).}
    \label{fig:rel-eff-delta-eps}
\end{figure}

\end{appendices}

\end{document}